\documentclass[11pt]{article}

\usepackage[a4paper, total={6in, 8in}]{geometry}

\usepackage{enumitem}
\usepackage{bbm,tikz}
\usetikzlibrary{cd}
\usepackage{algorithm}
\usepackage{algpseudocode}
\usepackage{float}
\algrenewcommand\algorithmicrequire{\textbf{Input:}}
\algrenewcommand\algorithmicensure{\textbf{Output:}}
\usepackage{amsthm,amssymb}
\usepackage{amsmath,bm}
\usepackage{parallel,enumitem}
\usepackage{comment}
\usepackage{todonotes}
\usepackage{xspace,setspace}
\usepackage[colorlinks=true, citecolor=blue]{hyperref}

\newcommand{\V}{\bm{V}}

\newcommand{\R}{\mathbb{R}}

\newcommand{\N}{\mathbb{N}}
\newcommand{\Q}{\mathbb{Q}}

\newcommand{\mL}{\mathcal{L}}

\newcommand{\Fb}{\mathbf{F}}
\newcommand{\cb}{\mathbf{c}}
\newcommand{\xb}{\mathbf{x}}
\newcommand{\yb}{\mathbf{y}}
\newcommand{\ab}{\mathbf{a}}
\newcommand{\hb}{\mathbf{h}}
\newcommand{\Pb}{\mathbf{P}}
\newcommand{\bb}{\mathbf{b}}
\newcommand{\gb}{\bm{g}}
\newcommand{\C}{\mathbb{C}}

\newcommand{\Cn}{\C^n}

\def\1{{\mathbbm 1}}
\newcommand{\gbt}{\widetilde{\gb}}

\newcommand{\SFX}{S_{(F,X)}}

\newcommand{\dg}{d_{\gb}}
\newcommand{\dF}[1][]{d_{F_{#1}}}
\newcommand{\dgbis}{d_{g}}
\newcommand{\degh}{d_{\hb}}

\newcommand{\compose}{\textsf{Compose}\xspace}

\newcommand{\InRadical}{\textsf{InRadical}\xspace}
\newcommand{\VectorBasis}{\textsf{VectorBasis}\xspace}
\newcommand{\InvariantSet}{\textsf{InvariantSet}\xspace}
\newcommand{\ComputeMatrix}{\textsf{ComputeMatrix}\xspace}
\newcommand{\InvariantSetBranch}{\textsf{InvariantSetBranch}\xspace}
\newcommand{\ComputeMatrixBranch}{\textsf{ComputeMatrixBranch}\xspace}
\newcommand{\TruncatedClass}{\textsf{TruncatedClass}\xspace}
\newcommand{\TruncatedIdeal}{\textsf{TruncatedIdeal}\xspace}
\newcommand{\TruncatedIdealBranch}{\textsf{TruncatedIdealBranch}\xspace}
\newcommand{\ParametricSol}{\textsf{ParametricSol}\xspace}
\newcommand{\CheckPI}{\textsf{CheckPI}\xspace}
\newcommand{\CheckPIBranch}{\textsf{CheckPIBranch}\xspace}
\newcommand{\coefficients}{\textsf{coefficients}\xspace}

\theoremstyle{plain}
\newtheorem{theorem}{Theorem}[section]
\newtheorem{lemma}[theorem]{Lemma}
\newtheorem{proposition}[theorem]{Proposition}
\newtheorem{corollary}[theorem]{Corollary}
\newtheorem{definition}[theorem]{Definition}
\theoremstyle{definition}
\newtheorem{example}{Example}
\theoremstyle{remark}
\newtheorem{remark}{Remark}

\newcommand{\programbox}[2][\linewidth]{
\begin{samepage}\normalfont
\vspace*{.5em}\hspace*{0.3cm}\fbox{
\hspace*{-0.3cm}\begin{minipage}{#1}
\vspace*{-.1em}\begin{algorithmic}
#2
\end{algorithmic}
\vspace*{-.2em}
\end{minipage}
}\\
\end{samepage}
}

\newcommand{\programboxappendix}[2][\linewidth]{
\begin{samepage}\normalfont
\vspace*{.5em}\hspace*{0.0cm}\fbox{
\hspace*{-0.3cm}\begin{minipage}{#1}
\vspace*{-.1em}\begin{algorithmic}
#2
\end{algorithmic}
\vspace*{-.2em}
\end{minipage}
}\\
\end{samepage}
}

\newcommand{\rev}[1]{{#1}}
\begin{document}
\title{Algebraic and Algorithmic Methods for Computing Polynomial Loop Invariants}

\author{Erdenebayar \textsc{Bayarmagnai} \textsuperscript{a}\\
\texttt{erdenebayar.bayarmagnai@kuleuven.be}
\and Fatemeh \textsc{Mohammadi} \textsuperscript{a,b}\\
\texttt{fatemeh.mohammadi@kuleuven.be}
\and R\'{e}mi \textsc{Pr\'ebet} \textsuperscript{c}\\
\texttt{remi.prebet@ens-lyon.fr}}

\date{\itshape\small
\textsuperscript{a}Department of Computer Science, KU Leuven,
Leuven,
Belgium\\
\textsuperscript{b}Department of Mathematics, KU Leuven,
Leuven,
Belgium\\
\textsuperscript{c}Inria, CNRS, ENS de Lyon, Université Claude Bernard Lyon 1, LIP, UMR 5668, 69342, Lyon cedex 07, France\\[.5em]
\today}

\maketitle
\vspace*{-2em}
\begin{abstract}
Loop invariants are properties of a program loop that hold both before and after each iteration of the loop. They are often used to verify programs and ensure that algorithms consistently produce correct results during execution. Consequently, generating invariants becomes a crucial task for loops. We specifically focus on polynomial loops, where both the loop conditions and the assignments within the loop are expressed as polynomials. Although computing polynomial invariants for general loops is undecidable, efficient algorithms have been developed for certain classes of loops. For instance, when all assignments within a while loop involve linear polynomials, the loop becomes solvable. In this work, we study the more general case, where the polynomials can have arbitrary degrees.  

Using tools from algebraic geometry, we present two algorithms designed to generate all polynomial invariants within a given vector subspace, for a branching loop with nondeterministic conditional statements. These algorithms combine linear algebraic subroutines with computations on polynomial ideals. They differ depending on whether the initial values of the loop variables are specified or treated as parameters. Additionally, we present a much more efficient algorithm for generating polynomial invariants of a specific form, applicable to all initial values. This algorithm avoids expensive ideal computations.  
\end{abstract}
%%%%%%%%%%%%%%%%%%%%%%%%%%%%%%%%%%%%%%%%%%%%%%%%%%%%%%%%%%%%%%%%%%%%%%%%
\section{Introduction}
%%%%%%%%%%%%%%%%%%%%%%%%%%%%%%%%%%%%%%%%%%%%%%%%%%%%%%%%%%%%%%%%%%%%%%%%

Loop invariants are properties that hold both before and after each loop iteration in a program. When a loop invariant takes the form of a polynomial equation or a polynomial inequality, it is called a polynomial invariant of the loop. In this paper, we focus exclusively on polynomial equation invariants, which we refer to simply as polynomial invariants. These play a crucial role in automating program verification, helping to ensure that algorithms consistently yield correct results during execution. 
In particular, several well-established safety verification methods, such as the Floyd-Hoare inductive assertion technique \cite{floyd1993assigning} and the termination verification via standard ranking functions technique \cite{manna2012temporal}, rely on loop invariants to verify correctness, enabling full automation of the verification process. 

In this paper, we specifically study polynomial loops, where the expressions in assignments are polynomials in program variables. Additionally, we consider only the case where the loop guard contains inequations, which are defined as the negation of an equation.  
\rev{Note here that allowing inequalities for invariants and guards would render the invariant verification problem more general than the positivity problem for linear recurrence
sequences, whose decidability status has been open for many decades (see \cite{kovacs2023algebra} and the references therein).}

More precisely, let $\ab=(a_1,\dotsc,a_n) \in \mathbb{C}^n$, and let $\hb = (h_1, \ldots, h_s)$ and $F = (f_1, \ldots, f_n)$ be two sequences of polynomials in $\mathbb{C}[x_1, \ldots, x_n]$, where the $x_i$'s represent \emph{program variables} and the $a_i$'s are the \emph{initial values} of the loop. Then, $\mL(\ab, \hb, F)$ denotes the loop: 
\begin{center}
  \programbox[0.5\linewidth]{
\State$(x_{1},\ldots, x_n)=(a_1,\ldots, a_n)$
\While{$h_1 \neq 0, \ldots, h_s\neq 0$}
\State $\begin{pmatrix}
x_1 \\
x_2 \\
\vdots \\
x_n
\end{pmatrix}
\xleftarrow{F}
\begin{pmatrix}
f_1\\
f_2\\
\vdots\\
f_n
\end{pmatrix}
$
\EndWhile
}\label{page:alg}
\end{center}
When no $\hb$ is considered, we write $\mL(\mathbf{a}, 1, F)$, as this corresponds to an infinite loop. Finally, we simply write $\mL$ when the corresponding loop is clear from the context.  
\medskip

In the following, consider a polynomial loop $\mL(\ab, \hb, F)$. We begin by defining the main object of interest to be computed in this paper.
\begin{definition}
The set $I_{\mL(\ab, \hb, F)}\subset \C[x_1,\ldots,x_n]$ of all polynomial invariants for $\mL(\ab, \hb, F)$ is called the \emph{invariant~ideal~of~$\mL$}. 
\end{definition}
The set defined above is thus called, as it is known (at least from \cite{rodriguez2004automatic}) to be an ideal of $\C[x_1, \ldots, x_n]$.  
Thanks to the Hilbert basis theorem, one can hope to compute a finite basis for this ideal to completely describe this invariant ideal.  
However, as recently shown in \cite{MMK2024}, this task is at least as hard as the Skolem problem, whose decidability has remained widely open for almost a century.  
Instead, in this paper, we aim to compute this invariant ideal partially, in the sense of the following definition.  
\begin{definition}  
Let $E$ be a \emph{finite}-dimensional vector subspace of $\C[x_1, \ldots, x_n]$. Then, the polynomial invariants of $\mL(\ab, \hb, F)$ in $E$  
are denoted by $I_{\mL(\ab, \hb, F),E}$.  
\end{definition}  
Note that this restricted set of polynomial invariants is a vector subspace of the finite-dimensional vector subspace $E$ of $\C[x_1, \ldots, x_n]$. Thus, we can reduce many routines to linear algebra problems and compute a vector basis for $I_{\mL(\ab, \hb, F), E}$.  
A classic choice for $E$ in the literature is the set of polynomials of degree bounded by some chosen constant. Making this degree large enough allows one to expect some polynomial invariant to be in $E$, while controlling the dimension of the problem.  

\medskip

\noindent\textbf{Related works.}
Over the past two decades, the computation of polynomial invariants for loops has been extensively studied, see e.g.~\cite{Unsolvableloops, de2017synthesizing, 
hrushovski2018polynomial,karr1976affine,
kovacs2008reasoning,
kovacs2023algebra,rodriguez2004automatic, 
rodriguez2007automatic, rodriguez2007generating, amrollahi2025solvable}.
For general loops, this problem is undecidable \cite{hrushovski2023strongest}. Therefore, special emphasis has been placed on certain families of loops, especially those in which the assertions are all linear or can be reduced to linear assertions.

A common approach for generating polynomial invariants entails creating a system of recurrence relations from a loop, obtaining a closed formula for this recurrence relation, and then computing polynomial invariants by removing the loop counter from the obtained closed formula (as in \cite{katz1976logical}). Note that it is straightforward to find such recursive formulas from a polynomial invariant. However, the reverse process is only feasible under very strong assumptions, as detailed in \cite{Unsolvableloops}. Specifically, one needs to identify polynomial relations among program variables, which is a challenging task. 

When each assignment within the loop is a linear function, Michael Karr introduced a pioneering algorithm for computing all linear invariants \cite{karr1976affine}. Subsequent studies, such as \cite{muller2004note} and \cite{rodriguez2007automatic}, have extended this algorithm to compute all polynomial invariants of bounded degree.
\rev{More recent approaches aim to compute the full invariant ideal. 
The algorithms of~\cite{hrushovski2018polynomial} and~\cite{klara2022computation} achieve this goal but suffer from prohibitive complexity: even the best ones remain multiple-exponential in the worst case. 
Nevertheless, restricted classes of loops admit substantially lower complexity. 
For instance, for non-branching loops with a single linear update, \cite{manssour2024simple} obtained a PSPACE algorithm, showing that efficient invariant generation is possible in suitably constrained settings.}

Another class of loops for which invariants have been successfully computed is the family of solvable loops. These loops are characterized by polynomial assignments that are either inherently linear or can be transformed into linear forms through a change of variables, as elaborated in \cite{de2016polynomial} and \cite{kovacs2008reasoning}. The techniques for generating all polynomial invariants of solvable loops are discussed in~\cite{de2016polynomial,rodriguez2004automatic}.
In~\cite{de2016polynomial}, solvable loops are transformed into linear loops, while in~\cite{rodriguez2004automatic} the authors rely on forward propagation and fixed-point computation. The methodology proposed in~\cite{kovacs2008reasoning} is specifically tailored for the special case of P-solvable loops, which is later extended in~\cite{humenberger2017invariant}.

Challenges persist when dealing with loops featuring non-linear or unsolvable assignments \rev{and nested nondeterministic branches}, as discussed in~\cite{Unsolvableloops, cachera2014inference, cyphert2024solvable, sankaranarayanan2004non, wang2024polynomial}. These methods generate polynomial invariants of a given degree; however, while they are sound, they do not guarantee completeness, meaning that invariants may be missed.
\rev{In particular, \cite{wang2024polynomial} searches for polynomial expressions of bounded degree that satisfy multi-dimensional C-finite recurrence relations. These expressions form a vector space, a viewpoint that conceptually parallels our own vector-space–based formulation.}

In \cite{muller2004computing}, Müller-Olm and Seidl employ ideas similar to ours in Algorithm~\ref{algo1}, with the notable difference that, through our geometric approach for computing the invariant set, we have established a better stopping criterion by comparing the equality of radical ideals rather than the ideals themselves. They also impose algebraic conditions on the initial values and subsequently compute polynomial invariants that must apply to all initial values satisfying these constraints, meaning some invariants applicable to all but finitely many initial values may be overlooked. 
In contrast, Algorithm~\ref{algo:class} yields polynomial invariants that depend on the initial values, addressing a much broader problem, which, to our knowledge, has not been previously tackled.  
Moreover, Algorithm~\ref{alg2}, which tackles cases with fixed initial values, is significantly faster than Algorithm~\ref{matrixalgo} and its counterpart in~\cite{muller2004computing}.~Additionally, Algorithm~\ref{alg3} generates all general invariants of a specific form whenever they exist, enabling us to generate invariants produced in~\cite{Unsolvableloops}.  

Finally, \cite{chatterjee2020polynomial} considers polynomial invariants as inequalities using tools like Putinar's Positivstellens\"atz.  However, this presents a different problem, 
involving semi-algebraic sets $X$ where $F(X) \subseteq X$. However, polynomial invariants do not necessarily satisfy this property.

\medskip
\noindent\textbf{Our contributions.} 
In this work, we consider the problem of fully generating the set $I_{\mL(\ab,\hb,F),E}$ of all polynomial invariants in a given vector space~$E$ for loops with polynomial updates of arbitrary degree and inequations as guards. 
In Section~\ref{sectioninvariantset}, we introduce invariant sets \rev{and present Algorithm~\ref{algo1} for computing them. In Section~\ref{sec: testinginvariants}, we build on this to develop Algorithm~\ref{algo:verify}, which verifies whether a given polynomial is an invariant.}
In Section~\ref{sectiongenerating}, using the results of Sections~\ref{sectioninvariantset} and~\ref{sec: testinginvariants}, we design two algorithms (Algorithms~\ref{matrixalgo} and~\ref{alg2}) to compute a basis of $I_{\mL(\ab,\hb,F),E}$, depending on whether the vector~$\ab$ is fixed or remains symbolic. These results and algorithms are further extended to branching loops with nondeterministic conditional statements.
In Section~\ref{sec:lift}, we introduce a specialized Algorithm~\ref{alg3}, which is significantly more efficient for computing all polynomial invariants of the form $f(x_1,\dotsc,x_n) - f(a_1,\dotsc,a_n)$. Notably, this algorithm avoids any computation with polynomial ideals. For example, while Algorithm~\ref{alg2} fails to compute all polynomial invariants of degree up to~5 within 360 seconds on the benchmarks \texttt{Fib1}, \texttt{Fib2}, and \texttt{Fib3}, Algorithm~\ref{alg3} successfully computes all invariants of this form up to degree~10 in under 120 seconds on the same benchmarks.
In Section~\ref{sec:termination}, we relate termination of loops with inequalities in their guards to polynomial invariants, deriving several necessary or sufficient conditions.  
Section~\ref{sec:implementation} presents an implementation of our algorithms together with experimental results.

\medskip
\noindent\textbf{Extended version.}
This paper is an extended version of the short paper \cite{bayarmagnai2024algebraic} published in the proceedings of ISSAC~2024. It provides full details, strengthens several results, and introduces multiple new contributions. Specifically, in Section~\ref{subsection:parinva}, we present Algorithm~\ref{algo:class}, which computes an explicit description of polynomial invariants that depend on the initial values. 
In Section~\ref{subsection:branching}, we extend the previously introduced algorithms to handle branching loops with nondeterministic conditionals and inequations in their guards; the earlier version addressed only single-path loops without guard conditions. 
In Section~\ref{sec:lift}, we considerably generalize \cite[Proposition~4.1]{bayarmagnai2024algebraic} and develop Algorithm~\ref{alg3}, which computes all polynomial invariants of a specific form, up to a given degree, for branching loops. This algorithm is independent of the methods in Section~\ref{sectiongenerating} and avoids expensive ideal computations. 
Finally, in Section~\ref{sec:implementation}, we compare the new Algorithm~\ref{alg3} with \textsf{Polar}~\cite{DBLP:journals/pacmpl/MoosbruggerSBK22} and provide new benchmarks obtained from an optimized implementation of Algorithm~\ref{alg2}. 
\rev{The following figure summarizes the dependencies among the algorithms. The core component is \InvariantSet\ (Algorithm~\ref{algo1}), and all other algorithms except Algorithm~\ref{alg3} rely on it.}

\begin{figure}[H]
\hspace*{-.5cm}
\begin{tikzcd}[column sep=4pt]
 & \textsf{InvariantSet}\;\text{(Alg.~\ref{algo1})} \arrow[dl] \arrow[d] & \\
 \textsf{CheckPI}\;\text{(Alg.~\ref{algo:verify})} \arrow[d] & \textsf{ComputeMatrix}\;\text{(Alg.~\ref{matrixalgo})} \arrow[dl] \arrow[d] & \\
  \textsf{TruncatedIdeal}\;\text{(Alg.~\ref{alg2})} & 
 \textsf{TruncatedClass}\;\text{(Alg.~\ref{algo:class})} & \textsf{SpecificForm}\;\text{(Alg.~\ref{alg3})}
\end{tikzcd}
\caption{Dependencies between Algorithms}
\end{figure}

%%%%%%%%%%%%%%%%%%%%%%%%%%%%%%%%%%%%%%%%%%%%%%%%%%%%%%%%%%%%%%%%%%%%%%%%
\section{The invariant sets}\label{sectioninvariantset}
%%%%%%%%%%%%%%%%%%%%%%%%%%%%%%%%%%%%%%%%%%%%%%%%%%%%%%%%%%%%%%%%%%%%%%%%
In this section, we present a method for identifying polynomial invariants by introducing invariant sets. As we will see later, these are naturally related to polynomial invariants.  
We denote the field of complex numbers by $\mathbb{C}$.
Throughout the paper, $\xb$ denotes $x_1, \dots, x_n$, and $\C[\xb]$ the polynomial ring in these variables.

\subsection{Properties of the invariant sets}
We start with the central technical object introduced in this paper.
\begin{definition}\label{def:invariantset}
Let $F : \Cn \longrightarrow \Cn$ be a map and $X$ be a subset of $\Cn$.
The invariant set $S_{(F,X)}$ of $(F,X)$ is defined as:
\begin{center}
    $S_{(F,X)} = \{x \in X \mid \forall m\in \N,\, F^{m}(x) \in X\},$
\end{center}
where $F^{0}(x)=x$ and $F^{m}(x)= F(F^{m-1}(x))$ for any $m>1$.
\end{definition}

Let us first establish an effective description of invariant sets. This will allow us to derive an algorithm for the particular case of algebraic varieties.

\begin{lemma}\label{lem:largeset}
     Let $X$ and $F$ be as above. Then $S_{(F,X)}$ is the largest $Y \subset X$ such that $F(Y)\subset Y$.
\end{lemma}
\begin{proof}
 Let $x \in S_{(F,X)}$. By definition, we have 
$   F^{m}(F(x)) \; = \; F^{m+1}(x) \in X$ for any $m\geq 0$.
Thus, $F(x) \in S_{(F,X)}$ for every $x \in S_{(F,X)}$,~i.e.~$F(S_{(F,X)}) \subset S_{(F,X)}$. 
 Conversely, let $Y \subset X$ such that $F(Y)\subset Y$. Thus, for any $m\geq 0$, we have 
     $F^{(m)}(Y)\subset Y\subset X$.
Hence, $Y\subset S_{(F,X)}$ which completes the proof.
\end{proof}

We now consider the case where the above $X$ and $F$ are respectively an algebraic variety and a polynomial map. Thanks to algebraic-geometry algorithmic tools, we can develop an effective method for incrementally computing invariant sets. First, we fix some terminology from algebraic geometry and refer to \cite{cox2013ideals, kempf_1993} for further details.

Let $S$ be a set of polynomials in $\mathbb{C}[\xb]$. Then, the \emph{algebraic variety} $\V(S)$ associated to $S$ is the common zero set in $\C^n$ of the polynomials in $S$. 
In particular, $\V(S) = \V(\langle S \rangle)$, where $\langle S \rangle$ is the ideal generated by~$S$.
Conversely, the \emph{defining ideal} of a subset $Y\subset\mathbb{C}^n$ is the set of polynomials in $\mathbb{C}[\xb]$ that vanish on $Y$. %; this is an ideal of $\mathbb{C}[x_1, \ldots, x_n]$.
The algebraic variety associated to the ideal $I(Y)$ is called the \emph{Zariski closure} of $Y.$ 
A map $F: \C^n \to \C^m$ is called a \emph{polynomial map}, if there exist $f_1,\ldots,f_m$ in $\C[\xb]$, such that $F(x) = (f_1(x),\ldots,f_m(x))$ for all $x \in \C^n$. For simplicity, we will henceforth refer to polynomial maps and their associated polynomials interchangeably.

In the following, we fix $F: \mathbb{C}^n \longrightarrow \mathbb{C}^m$ to be a polynomial map and $X \subset \mathbb{C}^m$ an algebraic variety.
 The following lemma is folklore in algebraic geometry. We provide a proof here as we could not find a suitable reference.
 %The following lemma is part of the folklore of algebraic geometry. We provide a proof here because we could not find a suitable reference for it.
\begin{lemma}\label{lem:equtionspreimage}
Assume that $X=\V(g_1, \ldots, g_k)$ and $F = (f_{1}, \ldots, f_{m})$.
%, where $g_1,\ldots,g_k\in \mathbb{C}[y_1, \ldots, y_m]$ and $f_1, \ldots, f_m \in \mathbb{C}[x_1, \ldots, x_n]$. 
Then 
\[
F^{-1}(X) = \V(g_{1}(f_1, \ldots , f_m), \ldots, g_{k}(f_1, \ldots, f_m)) \subset \mathbb{C}^n. 
\]
In particular, $F^{-1}(X) \subset \C^n$ is an algebraic variety.
\end{lemma}
\vspace{-2mm}
\begin{proof}
%Let $F = (f_{1}, \ldots, f_{m})$ and $X = \V(g_1, \ldots, g_k)$ as in the statement, 
By definition,
  \begin{align*}
    F^{-1}(X) %&= %\{ x \in \C^n \mid(f_{1}(x), \ldots, f_{m}(x)) \in \V(g_1,\ldots,g_k) \}\\
 &= \{ x \in \C^n \mid \forall\, 1 \leq i \leq k,\, 
 g_i(f_1(x),\ldots,f_m(x)) = 0\}.
  \end{align*}
 Let $h_i =  g_i(f_1,\ldots,f_m)\in\C[\xb]$, for all $1 \leq i \leq k$. Then
 $
    F^{-1}(X) = \V(h_1,\ldots,h_k),
 $ 
 and so $F^{-1}(X)$ is an algebraic variety.
\end{proof}
We now present an effective method to compute invariant sets of polynomial maps on algebraic varieties using an iterative outer approximation that converges to the goal. The method is effective because the sequence eventually stabilizes, which is easy to detect.
%by means of a stopping criterion for the intersection of the iterated preimages. 
We will denote $(F^{m})^{-1}$~as~$F^{-m}$.  
\begin{proposition}\label{prop:stabilization}
%Let $F : \Cn \longrightarrow \Cn$ be a polynomial map and $X\subseteq\Cn$ an algebraic variety. 
Let $X_0 = X$ and, for all $m \in \N$, let $X_m=\bigcap\limits_{i=0}^{m}F^{-i}(X).$\\[-.5em] 
Then, the following hold:
\begin{itemize}
        \item[$(a)$] $\SFX \subset X_{m+1} \subset X_{m}$ for all $m\geq 0$;
        \item[$(b)$] there exists %a natural number 
        $N\in\mathbb{N}$ such that $X_{N} = X_{m}$ for all $m > N$;
        \item[$(c)$]  if $X_{N}=X_{N+1}$ for some $N$, then $X_{N}=X_{m}$ for all $m> N$;
        \item[$(d)$]  the invariant set $S_{(F,X)}$ is exactly $X_{N}$.
    \end{itemize}
\end{proposition}
\begin{proof}
 $(a)$ The second inclusion is straightforward from the definition as 
 $$X_{m+1}=X_m\cap F^{-(m+1)}(X)\subseteq X_m.$$
 We now proceed to prove that $S_{(F,X)}\subseteq X_{m}$ for all $m$, by induction. 
By definition, $S_{(F,X)}$ is a subset of $X=X_{0}$ which proves the base case.
Now, let $m>0$ and assume that $S_{(F,X)}\subseteq X_{m-1}$.
%\textbf{Induction step.} 
Then, by Lemma~\ref{lem:largeset},
$$F(S_{(F,X)})\subset S_{(F,X)}\subset X_{m-1}.$$
%Therefore, $S_{(F,X)}$ is a subset of $F^{-1}(X_{m-1})$. 
%Note that $S_{(F,X)}$ is a subset of $X$ by the definition. 
Hence, by the induction assumption 
$S_{(F,X)}\subset F^{-1}(X_{m-1})\cap X_0 = X_{m}.$
%In particular, when $m=N$, we have that $S_{(F,X)}\subseteq X_{N}$.
 
        \smallskip
    
 \noindent $(b)$ From $(a)$, we have the following descending chain 
        \[ 
            X_{0} \supseteq \cdots\supseteq X_{m} \supseteq X_{m+1} \supseteq \cdots,
        \]
       which are algebraic varieties by Lemma~\ref{lem:equtionspreimage}. 
   Then, by \cite[Proposition 1.2]{hartshorne2013algebraic}, % the ascending chain of ideals: 
       \[
        I(X_{0})\subseteq \cdots \subseteq I(X_{m})\subseteq I(X_{m+1})\subseteq \cdots.
        \]
    Since $\C[x_{1},\ldots,x_{n}]$ is a Noetherian ring, there exists $N \in \mathbb{N}$ such that $I(X_{N})=I(X_{m})$ for all $m> N$. Therefore, 
    $$X_N = \V(I(X_N)=\V(I(X_m))=X_m \text{ for all } m\geq N.$$

\smallskip 
\noindent $(c)$ For such an integer $N$, the following allows us to conclude directly
\[
    X_{N+2} =X\cap F^{-1}(X_{N+1})=X\cap F^{-1}(X_N)=X_{N+1}.
 \] 

\smallskip
\noindent $(d)$ 
Let $N$ be as above and $x\in X_{N}$. Then, $ F^{N}(F(x)) = F^{N+1}(x) \in X $, since $X_{N}=X_{N+1}$.
Hence, $F(X_N) \subset X_N$, and as $\SFX$ is contained in $X_N$, by $(a)$, the inclusion is an equality by Lemma~\ref{lem:largeset}.
\end{proof}

\begin{remark}\label{remark1}
  By Theorem~\ref{prop:stabilization}.$(d)$, the invariant set $S_{(F,X)}$ is an algebraic variety, since each $X_{i}$ is an algebraic variety.
\end{remark}

\subsection{Algorithm for computing the invariant sets}
We now present Algorithm~\ref{algo1} for computing the invariant set associated with an algebraic variety and a polynomial map described by sequences of multivariate polynomials. We restrict ourselves to \emph{rational} coefficients, since this covers the target applications, and we need to work in a computable field for effectiveness. 
\rev{All statements in this section, including those on arithmetic complexity, remain valid over any computational field.}
We first define the subroutines involved in this algorithm.

\begin{itemize}
\item the procedure \compose takes as input two sequences of polynomials $\gb=(g_1,\dotsc,g_k)$ and $F=(f_1,\dotsc,f_n)$ in $\Q[\xb]$ and outputs a sequence of polynomials $(h_1,\dotsc,h_k)$, such that $h_i = g_i(f_1,\dotsc,f_n)$ for all $i$.
\item the procedure \InRadical takes as input a sequence $\gbt$ and a set $S$ both in $\Q[\xb]$ and decides if all the polynomials in $\gbt$ belong to the \emph{radical} of the ideal generated by $S$. This can be done following, for example, \cite[Chap. 4, \S 2, Proposition 8]{cox2013ideals}.
\end{itemize}

\begin{algorithm}
\caption{\InvariantSet}\label{algo1}
\begin{algorithmic}[1]
\Require Two sequences $\gb$ and $F = (f_1,\ldots, f_n)$ in $\mathbb{Q}[\xb]$.
\Ensure List of polynomials whose common zero-set is $S_{(F,{\V(\gb)})}$.
\State $S \gets \{\gb\};$
\State\label{InSStep2} $\gbt \gets \compose(\gb,\,F);$
\While{ $\InRadical(\gbt ,\,S)==\texttt{False}$}
\State $S \gets S \cup \{\gbt\};$
\State\label{InSStep5} $\gbt \gets \compose(\gbt,\,F);$
\EndWhile
\State \Return $S$;
% \State \Return N
\end{algorithmic}
\end{algorithm}

\smallskip
We now prove the termination and correctness of Algorithm~\ref{algo1}.

\begin{theorem}\label{thm:corralgo1}
    On input two sequences $\gb=(g_1,\ldots, g_k)$ and $F=(f_1,\ldots,f_n)$ of polynomials in $\Q[\xb]$, Algorithm~\ref{algo1} terminates and outputs a sequence of polynomials whose vanishing set is the invariant set $S_{(F, \V(g_1,\ldots, g_k))}$. %\textcolor{red}{What is $J$?}
\end{theorem}
\begin{proof}
%Let $V=\V(\gb)$ 
%and $F$ be the polynomial map$F=(f_1,\ldots,f_n):\Cn\longrightarrow \Cn$. L
Let $S_0 = \gb$, and for $m \geq 1$, let $S_m$ be the set contained in $S$ after completing $m$ iterations of the \textbf{while} loop. Similarly, let $\gbt_0 = \gb$, $\gbt_1 = \gb(F)$, and $\gbt_{m+1}$ be the sequence contained in $\gbt$ after $m$ iterations. 

Let $m \geq 0$, and let $X_m$ be as in Proposition~\ref{prop:stabilization}.  %We first prove that $X_{m}=\V(S_{m})$ for every $m\geq 0$.  
By construction,
$
S_{m} \,=\, \left\{\gbt_0,\dotsc, \gbt_{m}\right\} $, that is $S_m\,=\,\left\{\gb, \gb(F),\ldots, \gb(F^{m})\right\}
$, and so by Lemma~\ref{lem:equtionspreimage},
\begin{equation*}
  % \begin{split}
    X_{m} \,=\, \bigcap\limits_{i=0}^{m}F^{-i}(\V(\gb)) \,=\, \bigcap\limits_{i=0}^{m}\V(\gb(F^{i}))
    \,=\, \V(S_{m}).
   % \end{split}
\end{equation*} 
%which implies that $X_m=\V(I_m)$ for any $m\geq 0$. 
By Proposition~\ref{prop:stabilization}.(b), there exists $N\in\N$ such that $X_{N}=X_{N+1}$, that is $\V(S_N)=\V(S_{N+1})$. %by above. 
This means that the polynomial 
$\gbt_{N+1} = \gb(F^{N+1})$ vanishes on $\V(S_N)$, or equivalently by the Hilbert's Nullstellensatz \cite[Chap 4, \S 1, Theorem 2]{cox2013ideals}, that $\gbt_{N+1}$ belongs to $\sqrt{I(S_N)}$.

Hence, Algorithm~\ref{algo1} terminates after $N$ iterations of the \textbf{while} loop and outputs $S_N$. In particular, by Proposition~\ref{prop:stabilization}.(d), $S_{(F,X)} = X_N = \V(S_N)$, which proves the correctness of Algorithm~\ref{algo1}.
\end{proof}

\begin{proposition}\label{prop.complexity}
Using the notation of Proposition~\ref{prop:stabilization}, let $N$ be the smallest integer such that $X_N = X_{N+1}$. Let $\dg$ and $\dF$ denote bounds on the degree of the polynomials in $\gb$ and $F$, respectively, which are the inputs to Algorithm~\ref{algo1}. Then, Algorithm~\ref{algo1} satisfies the following:
\begin{itemize}
    \item[$(a)$] the {\normalfont{\bfseries while}} loop terminates after exactly $N$ iterations;
    \item[$(b)$] it performs at most $k\cdot \left(\dg\cdot \dF^N\right)^{O(n^2)}$ arithmetic operations in $\Q$.
\end{itemize}
\end{proposition}
\begin{proof}
The claim in $(a)$ was proved in the proof of Theorem~\ref{thm:corralgo1}. 

Now, let us prove $(b)$. To estimate the arithmetic complexity of Algorithm~\ref{algo1} in terms of the number of iterations of the loop, we need to evaluate the cost of one iteration. As before, let $\gbt_0 = \gb$, $\gbt_1 = \gb(F)$, and $\gbt_{m+1}$ be the sequence contained in $\gbt$ after $m$ iterations. We saw that for all $m \geq 0$, after $m$ iterations, the variable $S$ contains
$S_m = \{\gbt_0, \dotsc, \gbt_m \}$.

Fix $m \geq 1$. First observe that
\[
d_m = \deg(\gbt_m) \;\leq\; \dg \cdot \dF^m,
\]
which also bounds the degrees of the polynomials in $S_m$. In particular, $(d_m)_m$ is a non-decreasing sequence.

Hence, using the naive algorithm for sparse multivariate multiplication (see, e.g., \cite{HL2013}), one can show that the call to \compose, with input $\gbt_m$ and $F$, requires at most
$k \cdot (\dg \cdot \dF^m)^{O(n)}$,
operations in $\mathbb{Q}$. It remains to bound the cost of the call to \InRadical, with input $\gbt_{m+1}$ and $S_m$.

 Let $g$ be in the radical of the ideal generated by $S_m$. According to \cite[Corollary 1.7]{Ko1988}, there exists $r \in \mathbb{N}$ and $f_1, \dotsc, f_{m} \in \mathbb{Q}[\xb]$ such that
$\textstyle{g^r = \sum_{i=1}^m \gbt_i f_i}$
with $r \leq d_m^n$ and $\deg(\gbt_i f_i) \leq (1 + \deg(g))d_m^n$. Hence, by fixing the degree of these $f_i$'s and considering their coefficients as unknowns, one can reduce the radical membership test to the existence of a solution to a linear system of equations. More precisely, testing if the polynomials in $\gbt_{m+1}$ belong to the radical of the ideal generated by $S_m$ is reduced to solving $k$ linear systems, each of size at most
$(\dg \cdot \dF^m)^{O(n^2)}$.

Again, using classic algorithms, constructing and solving such a system can be done in a number of arithmetic operations polynomial in its size. In conclusion, the number of arithmetic operations performed during the $m$-th iteration can be bounded by
$k \cdot (\dg \cdot \dF^m)^{O(n^2)}$.

Summing over all iterations, one obtains the claimed result. Indeed, the cost of the final iteration dominates all preceding ones.
%Summing over all iterations, one obtains the claimed result, by the first item of the proposition. Indeed, the cost of the last iteration dominates all preceding ones.
\end{proof}

\begin{remark}
In practice, the radical membership test is performed using Gröbner basis algorithms, as outlined in \cite[Chap 4, \S 2, Proposition 8]{cox2013ideals}. While the worst-case complexity for these algorithms is doubly exponential in the number of variables \cite{mayr1982complexity}, these bounds are only reached in very specific cases. In practice, as discussed in \cite[\S 21.7]{von2013modern}, these algorithms are efficient and benefit from ongoing research \cite{eder2017survey} and efficient implementations \cite{berthomieu2021msolve}.
\end{remark}

\begin{remark}
The complexity estimate in Proposition~\ref{prop.complexity} is partial, as it depends on the quantity $N$, which is intrinsic to the algorithm and may not satisfy reasonable bounds with respect to the input size.
In the general case, the best bound in the literature is given in \cite[Theorem 6]{NY1999}, which bounds the length of the strictly descending chain of algebraic varieties defined by polynomials of bounded degree. This is the geometric counterpart of \cite{Se1972}. However, these bounds exhibit growth behavior similar to Ackermann's function (which is not primitive recursive) and have been proven to be sharp in \cite{MorenoSocias_1992}. We also refer to \cite{pastuszak2020ascending} for a more recent treatment of this problem.
However, under certain assumptions on the polynomial map $F$, primitive recursive bounds on $N$ and the complexity of Algorithm~\ref{algo1} can be derived. See \cite[Theorem 5]{NY1999} for an example.

We conclude by noting that these worst-case bounds are rare in our applications. The experimental section demonstrates that the algorithm is practically applicable to loops from the literature.
\end{remark}

\vspace{-2mm}
\section{Testing polynomial invariants}\label{sec: testinginvariants}
Now that we have an algorithm to compute invariant sets, we provide a criterion for identifying polynomial invariants using the invariant set, based on carefully chosen data. This will lead to an algorithm for checking whether a given polynomial is invariant.
Since the guard $[h_1 \neq 0, \ldots, h_s \neq 0]$ is equivalent to $[h_1 \times \cdots \times h_s \neq 0]$, we can consider polynomial loops with a single inequation in the guard condition. Recall that $I_{\mL(\ab, h, F)}$ denotes the set of all polynomial invariants of $\mL(\ab, h, F)$.
\begin{proposition}\label{prop:verify}
Let $h,g$ and $F=(f_1,\ldots, f_n)$  be polynomials in $\C[\xb]$.  Let z be a new indeterminate and $F_0(\xb,z)=(F(\xb),zh(\xb))$. Let $\ab \in \Cn$ and 
 $ X = \V(zg)\subset \C^{n+1}$.
Then, $g(\xb)\in I_{\mL(\ab,h, F)}$ if, and only if, $(\ab,1)\in S_{(F_0, X)}$.
\end{proposition}
\begin{proof}
    Let $\ab_0=\ab$ and $\ab_l=F(\ab_{l-1})$, for $l\geq 1$. 
Note that by construction,    
        $F_0^l(\ab, 1)=(\ab_l,h(\ab_0)\cdot\ldots\cdot h(\ab_{l-1}))$,
and thus 
    \begin{equation}\label{eqn:zgF}
        (zg)\circ F_{0}^l(\ab,1)=h(\ab_0)\cdot\ldots\cdot h(\ab_{l-1})g(\ab_l).
    \end{equation}
    Let $k\in \N \cup \{+\infty\}$ be the number of iteration after which $\mL(\ab, h,F)$ terminates. 
    First, assume that $g(\xb)\in I_{\mL(\ab, h,F)}$. 
    Then, 
    \begin{center}
        for $1\leq l < k+1$,\; $g(\ab_l)=0$\quad
    and, if $k<+\infty$, $h(\ab_k)= 0$. 
    \end{center}
    Hence, $F_0^{l}(\ab,1)\in X$, for any $l \geq 0$, by \eqref{eqn:zgF}, that is $(\ab,1)\in S_{(F_{0},X)}$. 
    
    Conversely, if $(\ab,1)\in S_{(F_0,X)}$ then, by \eqref{eqn:zgF}, 
        $h(\ab_0)\cdot\ldots\cdot h(\ab_{l-1})g(\ab_l)=0$
for all $l\geq 0$.   
    This means that $g(\ab_l)=0$ for all $l\leq k$, as $h(\ab_l)\neq 0$ by definition of $k$.
    In other words, $g(\xb)=0$ is a polynomial invariant of $\mL(\ab, h,F)$.
    \end{proof}

Algorithm~\ref{algo:verify} below, which checks whether a given polynomial is invariant, directly follows from Proposition~\ref{prop:verify}.
\begin{algorithm}
\caption{\CheckPI}\label{algo:verify}
\begin{algorithmic}[1]
\Require 
\begin{itemize}
    \item[] \hspace*{-.9cm}$g$, $\hb= (h_1,\ldots, h_k)$ and $F = (f_1,\ldots, f_n)$ in $\mathbb{Q}[\xb]$;\vspace*{-.3em}
    \item[] \hspace*{-.1cm}$\ab=(a_1,\ldots, a_n) \in \Q^n$.
\end{itemize}\vspace*{.2em}
\Ensure \texttt{True} if $g\in I_{\mL(\ab,\hb,F)}$; \texttt{False} else.
\State $h \gets h_1\cdot\ldots\cdot h_k;$
\State\label{CheckPIStep2} $F_0\gets (F,zh)$;
\State\label{CheckPIStep3} $\{P_1,\ldots, P_m\} \gets \InvariantSet(zg, F_0);$
\If{$P_1(\ab,1)=\cdots =P_m(\ab,1)=0$}
\State\Return \texttt{True}; 
\Else 
\State\Return \texttt{False};
\EndIf
\end{algorithmic}
\end{algorithm}
\begin{theorem}\label{thm:verify}
On input a sequence $\ab = (a_1, \ldots, a_n) \in \Q^n$, polynomials $g$, $\hb = (h_1, \ldots, h_k)$, and $F = (f_1, \ldots, f_n) \in \C[\xb]$, Algorithm~\ref{algo:verify} outputs {\normalfont\texttt{True}} if $g \in I_{\mL(\ab, \hb, F)}$, and {\normalfont\texttt{False}} otherwise.

Moreover, if $\dF, \degh$ and $\dgbis$ bound the degrees of respectively $F$, $\hb$ and $g$, then the number of required arithmetic operations in $\Q$ is bounded by
\[
    \left(\max\{\dF, k\degh\}^N\cdot \dgbis\right)^{O(n^2)},
\]
where $N$ is defined in Proposition~\ref{prop.complexity}, as the number of loop iterations performed in the call to \InvariantSet.
\end{theorem}
\begin{proof}
    Let $X =V(zg)$. By Theorem~\ref{thm:corralgo1}, the invariant set $S_{(F_0,X)}$ is the vanishing set of polynomials $P_1,\ldots, P_m$. Thus, by Proposition~\ref{prop:verify}, $g=0$ is a polynomial invariant of $\mL$ if and only if $P_1(\ab,1)=\cdots=P_m(\ab,1)=0.$
%which proves the correctness of Algorithm~\ref{algo:verify}.

Besides, the arithmetic complexity of the algorithm in dominated by the call to \InvariantSet, on input $zg$ and $F_0=(F,zh)$. Then, according to Proposition~\ref{prop.complexity}, one can directly conclude from the straightforward bounds on the degrees of these inputs.
\end{proof}

\begin{remark}
Algorithm~\ref{algo:verify} is essentially an ad hoc application of Algorithm~\ref{algo1}, and therefore exhibits a similar computational behavior. 
The decision problem tackled by Algorithm~\ref{algo:verify} is closely related to the Zeroness Problem for polynomial automata, as studied in~\cite{BDSW2017}.
Very roughly, a polynomial automaton is a machine that reads words over a finite alphabet and updates an internal numerical state by applying a fixed polynomial map for each input symbol. 
After the entire word has been processed, another polynomial function is applied to the final state to produce the output. 
The Zeroness Problem asks whether this output is always zero, regardless of the word read by the automaton. 
This decision problem is computationally extremely difficult: \cite{BDSW2017} shows that it is complete for a very high complexity class, with running times comparable to Ackermann’s function.

Our verification problem is conceptually similar. 
Algorithm~\ref{algo:verify} explores all possible sequences of polynomial updates (this is even more apparent in branching loops) and must determine whether every such sequence inevitably leads to a configuration in which the candidate invariant vanishes.
Nevertheless, under suitable restrictions on the polynomial updates, \cite[Theorem~5]{BDSW2017} and \cite[Theorem~5]{NY1999} show that the Zeroness Problem admits primitive-recursive upper bounds. 
The same type of assumptions on $F$ would give Algorithm~\ref{algo:verify} a comparable upper bound.
\end{remark}

\begin{example}\label{example1}
The following linear loop $\mathcal{L}(\ab,1,F)$ is taken from \cite{hrushovski2018polynomial}. In this example, we omit inequations in the guard for clarity in the output.

\programbox[0.55\linewidth]{
\State$(x_1, x_2)=(a_1,a_2)$
\While{true}
\State $\begin{pmatrix}
x_1 \\
x_2
\end{pmatrix}
\xleftarrow{F}
\begin{pmatrix}
10x_1-8x_2\\
6x_1-4x_2
\end{pmatrix}
$
\EndWhile
}

\noindent Let us check using Algorithm~\ref{algo:verify} whether the following polynomial 
 $$g=x_1^2-x_1x_2+9x_1^3-24x_1^2x_2+16x_1x_2^2$$
is an invariant of $\mathcal{L}(\ab,1,F)$, when $(a_1,a_2)=(0,1)$.
 Let $X= \V(g)$. 
    First, $\InvariantSet(g,F)$ computes the invariant set of $F=(10x_1-8x_2, 6x_1-4x_2)$ and $X$ through the following steps:
\begin{itemize}
\item initially, $S$ is set to $\{g\}$, and $\tilde{g}=\compose(g,F)$ that is
$$\tilde{g}=360x_1^3-1248x_1^2x_2+40x_1^2+1408x_1x_2^2-72x_1x_2-512x_2^3+32x_2^2;$$
\item computing a Gr\"obner basis for the ideal generated by $g$ and $1-t\tilde{g}$, the call $\InRadical(\tilde{g}, S)$ returns \texttt{False};
\item the set $S$ is then updated to $\{g,\tilde{g}\}$ %include both $g$ and $\tilde{g}$
%, resulting in 
%S=S\cup\{\tilde{g}\}=\{x_1^2-x_1x_2+9x_1^3-24x_1^2x_2+16x_1x_2^2,\, 360x_1^3-1248x_1^2x_2+40x_1^2+1408x_1x_2^2-72x_1x_2-512x_2^3+32x_2^2\}$, 
and $\tilde{g}$ is recomputed as $\compose(\tilde{g},F)$:
\[
7488x_1^3-26880x_1^2x_2+832x_1^2+31744x_1x_2^2-1600x_1x_2-12288x_2^3+768x_2^2;
\]
\item finally one checks that $\InRadical(\tilde{g}, S)$ yields \texttt{True}.
\end{itemize}
Thus, $\InvariantSet(g,F)$ outputs the first two computed polynomials $g(x_1,x_2), $ and $g(F(x_1,x_2))$, whose common zero set is then $S_{(F,X)}$.
%\begin{flushleft}
%$\{x_1^2-x_1x_2+9x_1^3-24x_1^2x_2+16x_1x_2^2\textbf{, }60x_1^3-1248x_1^2x_2+40x_1^2+1408x_1x_2^2-72x_1x_2-512x_2^3+32x_2^2\}.$
%\end{flushleft}
%This set represents the polynomials iteratively generated by the algorithm. 
Since %g(0,1)=0$ but 
$g(F(0,1))=-480$, on input $g$ and $F$, the output of Algorithm~\ref{algo:verify} is \texttt{False}. Therefore, $g$ is not a polynomial invariant of $\mL((0,1),1, F)$.
\end{example}

\vspace{-4mm}
%%%%%%%%%%%%%%%%%%%%%%%%%%%%%%%%%%%%%%%%%%%%%%%%%%%%%%%%%%%%%%%%%%%%%%%%
\section{Generating polynomial invariants}\label{sectiongenerating}
%%%%%%%%%%%%%%%%%%%%%%%%%%%%%%%%%%%%%%%%%%%%%%%%%%%%%%%%%%%%%%%%%%%%%%%%
In this section, we present various algorithms for generating polynomial invariants, depending on whether the initial values of the loops are specified or treated as variables, and we extend our results to branching~loops.
%We will denote by ${\rm span}(v_1 ,\ldots, v_n)$ the vector space generated by $v_{1},\ldots, v_{n}$. 
%We use bold letters to denote vectors (for instance $\bf x_0$, $\bf a_{i_1,\ldots i_n}$).

\subsection{Loops with parametric initial values} \label{subsection:parinva}
We begin with a criterion for identifying polynomial invariants in a vector subspace $E \subseteq \C[\xb]$ specified by its generators. This extends the criterion of Proposition~\ref{prop:verify} to an ansatz for such invariants. \rev{We denote by $I_{\mathcal{L}(\ab,h,F),E}$ the set of all polynomial invariants in $E$ for the loop $\mathcal{L}(\ab,h,F)$.}

\begin{proposition}\label{prop3.4}
Let $F = (f_1, \ldots, f_n)$ and $h$ be polynomials in $\C[\xb]$.
Let $E$ denote the vector subspace of $\C[\xb]$ spanned by $g_1(\xb), \ldots, g_m(\xb)$.
Let $\yb=(y_1,\dotsc,y_m)$ be new indeterminates and
\[
 g(\xb,\yb)=y_1g_1(\xb)+\cdots +y_mg_m(\xb)  \in \C[\xb,\yb].
\] 
Additionally, let $z$ be a new indeterminate, and define $X = V(zg) \subset \C^{n+m+1}$ and $F_m(\xb, \yb, z) = (F(\xb), \yb, zh(\xb))$.  
Then, for any $\ab \in \C^n$,  
\[
I_{\mathcal{L}(\ab, h, F), E} = \big\{ g(\xb, \mathbf{b}) \mid (\ab, \mathbf{b}, 1) \in S_{(F_m, X)} \big\}.
\]  
\end{proposition}
\begin{proof}
    Let $G(\xb,\yb)=(F(\xb),\yb)$ and fix $\ab \in \C^n$.
    Note that $g(G^l(\xb,\yb))=g(F^l(\xb),\yb)$, so for any $\bb\in\C^m$ we have
    $$g(\xb,\bb)\in I_{\mL(\ab, h, F)} \iff g(\xb,\yb)\in I_{\mL((\ab,\bb), h,G)}.$$
By Proposition~\ref{prop:verify}, the latter is equivalent to  
$(\ab,\bb,1)\in S_{(F_m,X)}$. %, as desired.
\end{proof}

The set of polynomial invariants of $\mL$ in $E$ is itself a finite-dimensional vector space, as is $E$. Consequently, we show below that $I_{\mL, E}$ can be parameterized by a system of linear equations whose coefficients depend polynomially on the initial values.

Let $\hb = (h_1, \ldots, h_k)$, $\gb = (g_1, \ldots, g_m)$, and $F = (f_1, \ldots, f_n)$ be polynomials in $\C[\xb]$, and let $E$ be the vector space generated by $\gb$. Below, we present the algorithm \textsf{ComputeMatrix}, which computes a matrix $A$ with polynomial entries such that for any $\ab \in \C^n$, Equation~\eqref{eq:matrixcon} below is satisfied.  
%\begin{equation}\label{eq:matrixcon}  
 %     I_{\mL(\ab, \hb, F), E} =  
 % \left\{ \sum_{i \leq m} b_{i} g_i(\xb) \mid (b_1, \ldots, b_m) \in \ker\,A(\ab) \right\},  
%\end{equation}  

In Algorithm~\ref{matrixalgo}, the procedure \textsf{Matrix} takes as input a sequence of polynomials $\widetilde{P}_1, \ldots, \widetilde{P}_N$ in $\Q[\xb, y_1, \ldots, y_m]$ such that $\widetilde{P}_1, \ldots, \widetilde{P}_N$ are linear in the variables $\yb$, and outputs a polynomial matrix $A$ with coefficients in $\Q[\xb]$ such that 
$[\widetilde{P}_1, \ldots, \widetilde{P}_N]^t = A \cdot [y_1, \ldots, y_m]^t$.

\begin{algorithm}
\caption{\ComputeMatrix}\label{matrixalgo}
\begin{algorithmic}[1]
\Require Three sequences $\gb = (g_1,\ldots, g_m),\ \hb=(h_1,\ldots, h_k)$ and $F = (f_1,\ldots, f_n)$ in $\mathbb{Q}[\xb]$.
\Ensure A polynomial matrix $A$ satisfying~(\ref{eq:matrixcon}).
\State $g\gets y_1g_1+\cdots +y_mg_m$;
\State $h\gets h_1\cdot \ldots \cdot h_k$;
\State\label{ComMatStep3} $F_m \gets (F,\yb,zh)$;
\State\label{ComMatStep4} $(P_1(\xb,\yb,z),\ldots ,P_N(\xb,\yb,z))\gets \InvariantSet(zg, F_m)$;
\State $(\widetilde{P}_1(\xb,\yb),\ldots , \widetilde{P}_N(\xb,\yb))\gets (P_1(\xb,\yb,1),\ldots ,P_N(\xb,\yb,1))$
\State $A\gets \textsf{Matrix}(\widetilde{P}_1(\xb,\yb),\ldots , \widetilde{P}_N(\xb,\yb))$;
\State \Return $A$;
% \State \Return N
\end{algorithmic}
\end{algorithm}

\begin{theorem}\label{algogeneral}
Let $F = (f_1, \ldots, f_n)$, $\hb = (h_1, \ldots, h_s)$, and $\gb = (g_1, \ldots, g_m)$ be polynomials in $\Q[\xb]$, and let $E$ denote the vector space spanned by $\gb$. Given $\gb$, $\hb$, and $F$ as input, Algorithm~\ref{matrixalgo} outputs a polynomial matrix $A$ such that for any $\ab \in \Cn$,
\begin{equation}\label{eq:matrixcon}  
      I_{\mL(\ab, \hb, F), E} =  
  \left\{ \sum_{i \leq m} b_{i} g_i(\xb) \mid (b_1, \ldots, b_m) \in \ker\,A(\ab) \right\},  
\end{equation}  
where $\ker\,A(\ab)$ is the right kernel of $A$, whose entries are evaluated at $\ab$.  

Moreover, if $\dg, \degh$ and $\dF$ are bounds on the degrees of respectively $\gb, \hb$ and $F$, then the number of required arithmetic operations in $\Q$ is at most
\[
    \left(\max\{\dF, k\degh\}^N\cdot \dg\right)^{O(n^2+m^2)},
\]
where $N$ is defined in Proposition~\ref{prop.complexity}, as the number of loop iterations performed in the call to \InvariantSet.
\end{theorem}
\begin{proof}
Let $y_1, \ldots, y_m$ be new indeterminates, $h = h_1 \cdot \ldots \cdot h_k$, and define $g$, $F_m$, and $X$ as in Proposition~\ref{prop3.4}. Then, by Proposition~\ref{thm:corralgo1}, on input $(zg, F_m)$, \InvariantSet computes polynomials 
$P_1, \ldots, P_N \in \Q[\xb, \yb, z]$
whose common vanishing set is $S_{(F_m, X)}$. Let $\widetilde{P}_j(\xb, \yb) = P_j(\xb, \yb, 1)$ for every $j \in \{1, \ldots, N\}$, then by the construction of \InvariantSet in Algorithm~\ref{algo1}, 
\begin{align*}
    \widetilde{P}_j &= P_j(\xb, \yb, 1) = (zg) \circ \Big(F^j(\xb), \; \yb, \; h(F^{0}(\xb)) \cdot \ldots \cdot h(F^{j-1}(\xb))\Big) \\
    &= h(F^{0}(\xb)) \cdot \ldots \cdot h(F^{j-1}(\xb)) \cdot g(F^{j}(\xb), \yb)
\end{align*}
Thus, the $\widetilde{P}_j$'s are linear in the $y_i$'s, and there exists a matrix $A$ of size $N \times m$ with coefficients in $\Q[\xb]$ such that 
   % \begin{equation*}
   $ \begin{bmatrix}
        \widetilde{P}_1\cdots
        \widetilde{P}_N 
    \end{bmatrix}^t
    =
    A \cdot 
    \begin{bmatrix}
        y_1\cdots
        y_m   
    \end{bmatrix}^t.$
%    \end{equation*}
Then, we are done for the correction, as by Proposition~\ref{prop3.4}, for any $\ab \in \C^n$,
\[
   I_{\mathcal{L}({\bf a}, h,F),E} = \big\{ g(\xb,{\bf b}) \mid \widetilde{P}_{1}(\ab,\bb)=\cdots =\widetilde{P}_{1}(\ab,\bb)=0 \big\}.%\tag*{\qed}
\]
The complexity estimate is straightforward, and similar to Theorem~\ref{thm:verify}, as the cost is dominated by the call to \InvariantSet.
\end{proof}
%\vspace{-1cm}

\begin{remark}
Note that for any polynomial matrix $A$ that satisfies~\eqref{eq:matrixcon}, any row-equivalent matrix $B$ also satisfies~\eqref{eq:matrixcon}. This provides some flexibility in potentially reducing the output of \ComputeMatrix.
\end{remark}

\begin{example}\label{ex2}
Consider the loop $\mathcal{L}(\ab,1,F)$ from Example~\ref{example1}. In \cite{hrushovski2018polynomial}, some polynomial invariants are computed for specific initial values to verify the non-termination of the linear loop with the guard $``2x_2 - x_1 \geq -2"$. In our analysis, we extend this validation by computing \emph{all} polynomial invariants up to degree 2 for \emph{arbitrary} initial values. Since $(\xb)_{\leq 2} = (1, x_1, x_2, x_1^2, x_1 x_2, x_2^2)$ is a basis for $\C[x_1, x_2]_{\leq 2}$, the input for Algorithm~\ref{matrixalgo} is then $(F, 1, (\xb)_{\leq 2})$.
In the following, we detail the execution of Algorithm~\ref{matrixalgo}.

\smallskip\noindent The first step consists of running \InvariantSet on input 
\begin{align*}
    F_6 &= (10x_1 - 8x_2, 6x_1 - 4x_2, y_1, \ldots, y_6, z) \; \text{ and} \\
    g &= z \cdot (y_1 + y_2x_1 + y_3x_2 + y_4x_1^2 + y_5x_1x_2 + y_6x_2^2),
\end{align*}  
where the $z$ and $y_i$'s are new variables. The output is five
polynomials $P_1, \dotsc, P_5$ in $\Q[x_1, x_2, y_1, \ldots, y_6, z]$
whose common zero set is $S_{(F_6, X)} \subset \C^8$. 
Finally, the polynomials $\widetilde{P}_{i} = P(\xb, \yb, 1)$, where $1 \leq i \leq 5$, define a linear system with unknowns $y_1, \ldots, y_6$, given by the following matrix $A$.

\vspace{1em}
\begin{center}
\noindent\scalebox{0.65}{%\[\scriptsize{\arraycolsep=0.5\arraycolsep
$
\begin{bmatrix}
        x_1^2 & x_1x_2 & x_2^2 & x_1&  x_2& 1\\
        (10x_1-8x_2)^2  &  (10x_1-8x_2)(6x_1-4x_2) & (6x_1-4x_2)^2 & 10x_1-8x_2& 6x_1-4x_2& 1\\
        (52x_1-48x_2)^2& (52x_1-48x_2)(36x_1-32x_2)& (36x_1-32x_2)^2& 52x_1-48x_2&  36x_1-32x_2&1\\
         (232x_1-224x_2)^2& (232x_1-224x_2)(168x_1-160x_2) & (168x_1-160x_2)^2  & 232x_1-224x_2& 168x_1-160x_2& 1\\
         (976x_1-960x_2)^2& (976x_1-960x_2)(720x_1-704x_2)& (720x_1-704x_2)^2 & 976x_1-960x_2&720x_1-704x_2&1
   \end{bmatrix}
$}
\end{center}
\end{example}

Now, we provide a novel explicit description of the set of polynomial invariants depending on initial values by solving the parametric linear system of equations output by \ComputeMatrix.
We rely on a method described in \cite{sit1992algorithm}, which incrementally constructs constructible sets in the parameter space, where the set of solutions can be explicitly computed. The following proposition is a reformulation of \cite[Theorem 4.1]{sit1992algorithm} in our context.

\begin{proposition}[{\cite[Theorem 4.1]{sit1992algorithm}}]\label{prop:ParSol}
  Let $A(\xb)$ be an $r \times m$ polynomial matrix whose entries lie in $\Q[\xb]$. Then, there exists an algorithm \ParametricSol that computes lists of polynomial equations and inequations defining constructible subsets $S_1, \ldots, S_k$ of $\C^n$ and matrices $Z_1(\xb), \ldots, Z_k(\xb)$ whose entries lie in $\Q(\xb)$, such that
     \begin{enumerate}
        \item $\C^n = S_1 \cup \ldots \cup S_k$ \; and \;$k \leq \binom{m+r}{r}$;
        \item for every $i \in \{1, \ldots, k\}$ and any $\ab \in S_i$, the columns of $Z_i(\ab)$ form a basis of $\ker A(\ab)$.
    \end{enumerate}
\end{proposition}
\begin{proof}
For clarity, we briefly describe an adapted version of the proof from \cite{sit1992algorithm} and the underlying algorithm.
For any $s \leq \min\{r, m\}$, choose a non-singular $s \times s$ submatrix $M(\xb)$ of $A$. 
Without loss of generality, write
\[
A(\xb) = 
\begin{bmatrix}
   M(\xb) & M'(\xb) \\
   N(\xb) & N'(\xb)
\end{bmatrix}\quad\text{and}\quad Z(\xb) = 
\begin{bmatrix}
    -K(\xb)M(\xb) \\
    I_{n-s}
\end{bmatrix},
\]
where $K(\xb)$ is the inverse of $M(\xb)$ and $I_{n-s}$ is the $(n-s) \times (n-s)$ identity matrix. Note that, since the inverse of $M(\xb)$ can be computed from its determinant and minors, the entries of $K(\xb)$ lie in $\Q(\xb)$. Let $S$ be the constructible set defined as the locus where all $s+1$ minors of $A$ vanish but $\det(M(\xb)) \neq 0$. Then, for any $\ab \in S$, the columns of $Z(\ab)$ form a vector basis of $\ker A(\ab)$. 
Taking each such $S$ and $Z$ for every possible square submatrix of $A(\xb)$, we obtain the output claimed in the proposition. The number of such choices, gives the claimed bound on $k$.
\end{proof}
 
In Algorithm~\ref{algo:class} below, we extend \ComputeMatrix by one step. For polynomials $h$ and $F = (f_1, \ldots, f_n)$ in $\C[\xb]$, we compute polynomials that define constructible sets $S_1, \ldots, S_k$ covering $\C^n$, along with lists of polynomials $T_1, \ldots, T_k \subset \C(\yb)[\xb]$, such that for each $i \in \{1, \ldots, k\}$ and any $\ab \in S_i$, %the set
\begin{equation}\label{classcond}
 \{f({\bf x},{\bf a})\mid f\in T_i\} \text{ is a vector basis for } I_{\mathcal{L}({\bf a},\hb, F), E}.    
 \end{equation}
forms a vector space basis for $I_{\mathcal{L}(\ab, h, F), E}$. 
%In Algorithm~\ref{algo:class}, 
We adopt the following notation:
\[
[g_1,\ldots,g_m]\cdot (Z^1,\ldots,Z^k):= ([g_1,\ldots,g_m] \times Z^1, \ldots,[g_1,\ldots,g_m]\times Z^k).
\]
\vspace{-5mm}
\begin{algorithm}
\caption{\TruncatedClass}\label{algo:class}
\begin{algorithmic}[1]
\Require Three sequences $\gb = (g_1,\ldots, g_m), \hb=(h_1,\ldots, h_k)$ and $F = (f_1,\ldots, f_n)$ in $\mathbb{Q}[\xb]$.
\Ensure Polynomials defining $S_1,\ldots, S_k$ covering $\Cn$ and lists of polynomials $T_1,\ldots, T_k\subset \C(\yb)[\xb]$ satisfying~\eqref{classcond}.
\State $A\gets \ComputeMatrix(\gb,\hb,F)$;
\State $((Q_1,\ldots,Q_k), (Z^1,\ldots,Z^k))\gets \ParametricSol(A)$;
\State $(T_1,\ldots, T_k)\gets [g_1,\ldots,g_m]\cdot (Z^1,\ldots,Z^k);$
\State \Return $((Q_1,\ldots,Q_k), (T_1,\ldots,T_k))$;
\end{algorithmic}
\end{algorithm}
\vspace{-3mm}
\begin{theorem}\label{thm:alg4}
Let $\gb=(g_1,\ldots, g_m), \hb=(h_1,\ldots, h_k)$ and $F=(f_1,\ldots, f_n)$ be polynomials in $\Q[\xb]$. Let $E$ be the vector space spanned by $\gb$. On input $\gb,\hb$ and $F$, Algorithm~\ref{algo:class} outputs polynomials defining constructible sets $S_1,\ldots, S_k$ of $\mathbb{C}^n$ and lists of polynomials $T_1,\ldots, T_k \subset \mathbb{C}(\yb)[\bf x]$ such that 
    \begin{enumerate}
        \item $\Cn=S_1\cup\cdots \cup S_k$;
        \item For every $i\in \{1,\ldots, k\}$ and any ${\bf a} \in S_i$, $\{f({\bf x}, {\bf a})\mid f\in T_i\}$ is a vector space basis for $I_{\mathcal{L}({\bf a},\hb, F), E}$.
    \end{enumerate}
Moreover, if $N$ is the number defined in Theorem~\ref{algogeneral}, then $k\leq \binom{m+N}{N}$.
\end{theorem}
\begin{proof}
   According to Theorem~\ref{algogeneral}, on input $(\gb,\hb,F)$, \ComputeMatrix computes a polynomial $N\times m$ matrix $A$ such that for any ${\bf a} \in \mathbb{Q}^n$,
    \[
  I_{\mathcal{L}({\bf a}, \hb, F ),E} = 
  \left\{ \sum_{ i\leq m}b_{i}g_i \mid (b_1,\ldots,b_m) \in \ker\,A(\ab) \right\}.
 \]
%From the previous line, $\displaystyle\sum_{ i \leq m}b_{i}g_i=0$ is a polynomial invariant for $\mathcal{L}(\ab, \hb, F)$ if and only if $A(\ab)(b_1,\ldots, b_m)^{T}=0$. 
By Proposition~\ref{prop:ParSol}, the columns of $Z_{i}(\ab)$ form a basis for $\ker A(\ab)$ for any $i$ and any $\ab\in S_i$, hence %for any $i$ and any $\ab\in S_i$, 
$\{f(\xb,\ab)\mid T_i\}$ is a basis for $I_{\mL(\ab, \hb, F),E}$.
\end{proof}

In the following example, we continue Example~\ref{ex2} to classify polynomial invariants up to degree 2 with respect to initial values.
\begin{example}\label{ex2cont}
From the output of Example~\ref{ex2}, we proceed by computing an explicit basis for the corresponding vector space of $I_{\mathcal{L},\Q[\xb]_{\leq 2}}$. Let $A$ be the matrix from Example~\ref{ex2}. First, we compute all the maximal minors of $A$ and observe that all are zero when $3x_1^2-7x_1x_2+4x_2^2=0$. Therefore, the dimension of $\ker A(a_1,a_2)$ is $1$ if and only if 
$3a_1^2-7a_1a_2+4a_2^2\neq 0.$
For these initial values, $\ker A(a_1,a_2)$ is generated by 
$$(0, (3a_1-4a_2)^2, -(3a_1-4a_2)^2, -9(a_1-a_2), 24(a_1-a_2), -16(a_1-a_2)).$$
We compute all the minors and find a basis for each case, leading us to identify the following four distinct cases based on the dimensions and minors. %These cases, denoted as conditions $Q_i$, are as follows:
For $(a_1, a_2)$ satisfying the conditions $Q_i$ for $i = 1, \ldots, 4$, 
\[
Q_1 = \{a_1 = 0, a_2 = 0\}, \quad Q_2 = \{a_1 = a_2, a_1 \neq 0\}, 
\]
\[
Q_3 = \{3a_1 = 4a_2, a_1 \neq 0\}, \quad Q_4 = \{3a_1 \neq 4a_2, a_1 \neq a_2\},
\]
the columns of the following matrices $Z^i$ form a vector basis for $\ker A(a_1, a_2)$:
{\small\[
Z^1 = \begin{bmatrix} 0 & 0 & 0 & 0 & 0 \\ 1 & 0 & 0 & 0 & 0 \\ 0 & 1 & 0 & 0 & 0 \\ 0 & 0 & 1 & 0 & 0 \\ 0 & 0 & 0 & 1 & 0 \\ 0 & 0 & 0 & 0 & 1 \end{bmatrix}, \quad
Z^2 = \begin{bmatrix} 0 & 0 & 0 \\ 1 & 0 & 0 \\ -1 & 0 & 0 \\ 0 & 1 & 0 \\ 0 & -1 & -1 \\ 0 & 0 & 1 \end{bmatrix},
\]
\[
Z^3 = \begin{bmatrix} 0 & 0 & 0 \\ 3 & 0 & 0 \\ -4 & 0 & 0 \\ 0 & -3 & 0 \\ 0 & 16 & -3 \\ 0 & -16 & 4 \end{bmatrix}, \quad
Z^4 = \begin{bmatrix} 0 \\ (3a_1 - 4a_2)^2 \\ -(3a_1 - 4a_2)^2 \\ -9(a_1 - a_2) \\ 24(a_1 - a_2) \\ -16(a_1 - a_2) \end{bmatrix}.
\]
}
Let $T_i$ be $[1,x_1,x_2,x_1^2,x_1x_2,x_2^2]\times Z^i$ for each $i$. %$ = 1, \ldots, 4$. 
The output of Algorithm~\ref{algo:class} is $((Q_1,Q_2, Q_3, Q_4), (T_1,T_2,T_3, T_4))$, which provides the following description:
%Therefore, we have the following description of polynomial invariants depending on initial values:
\begin{center}
\scalebox{0.75}{\begin{tabular}{ |c|c| } 
\hline
Initial values & Basis of $I_{2,\mathcal{L}}$\\

\hline
$S_1=\{(0,0)\}$ & $T_1=\{x_1,x_2,x_1x_2,x_1^2,x_2^2\}$ \\ \hline
$S_2= \{(a,a)\mid a\in \C^*\}$ &  $T_2=\{x_1-x_2, x_1^2-x_1x_2, -x_1x_2+x_2^2\}$  \\ \hline
$S_3=\left\{\left(\dfrac{4}{3}a,a\right)\mid a\in \C^*\right\}$ & $T_3= \{3x_1-4x_2, -3x_1^2+16x_1x_2-16x_2^2, -3x_1x_2+4x_2^2\}$\\
\hline
$S_4 = \left\{(a_1,a_2)\in \C^2\mid a_1\neq\frac{4}{3}a_2, a_1 \neq a_2\right\}$ & $T_4=\{(3a_{1}-4a_{2})^{2}x_1-(3a_{1}-4a_{2})^{2}x_2-9(a_1-a_2)x_1^2$\\ 
  &$+24(a_1-a_2)x_1x_2-16(a_1-a_2)x_2^2\}$
\\
\hline
\end{tabular}}
\end{center}
It is noteworthy that in the first three cases, the truncated invariant ideal is independent of the initial values. This occurs because these cases correspond to degenerate situations where the initial values are non-generic, i.e., they lie on a proper algebraic variety within $\mathbb{C}^2$. In contrast, the last case is generic, and the output depends on the initial values.
\end{example}

\subsection{Loops with given initial value}\label{subsec:initialvalue}
Although the algorithm outlined in Theorem~\ref{algogeneral} addresses the most general case, in practice, it quickly becomes impractical, even for small inputs.~In this section, we focus on the case where the initial values of the loops are fixed and propose a more efficient adapted algorithm.

The following proposition presents a sufficient condition for a polynomial to be an invariant, based on the loop's fixed initial values.
\begin{proposition}\label{prop_sufficient}
Consider a loop $\mathcal{L}(\ab_0,h,F)$, where $\ab_0 \in \C^n$ is fixed.~Given polynomials $g_1(\xb), \ldots, g_m(\xb)$ in $\C[\xb]$, define $G_0 = y_1g_1(\xb) + \cdots + y_mg_m(\xb)$. For $k \geq 1$, define
\vspace{-5mm}
\[
G_k(\xb, \yb) = \prod_{j=0}^{k-1} h(F^j(\xb)) \cdot \sum_{i=1}^{m} y_i g_i(F^k(\xb)).
\]
Let $\bb \in \C^m$. If $G_0(\xb, \bb)$ is a polynomial invariant of $\mathcal{L}(\ab_0, h, F)$, then $\bb$ satisfies the following system of linear equations in the $y_i$'s:
\[
G_0(\ab_0, \yb) = \cdots = G_K(\ab_0, \yb) = 0 \quad \text{for any } K \geq 0.
\]
\end{proposition}
\noindent This proposition is a direct consequence of the following lemma.
\begin{lemma}\label{prop4}
Let $h,\ \gb = (g_1,\ldots, g_m)$ and $F = (f_1,\dotsc,f_n)$ be polynomials in $\C[\xb]$.
Let 
\[
    X = \V\Big(z\cdot (y_1g_1(\xb)+\cdots +y_mg_m(\xb))\Big)\subset \mathbb{C}^{n+m+1}.
    \vspace*{-.3em}
\]
For $k\geq 0$, let $X_{k} = \cap_{j=0}^{k}F_{m}^{-j}(X)$
and $S_{k}=X_{k}\cap \V(\xb-\ab_0,z-1)$, where\\[.5em] $\ab_0\in \Cn$.
%, land let
%$$
%X_{k} = \displaystyle\bigcap_{j=0}^{k}F_{m}^{-j}(X) 
%\text{\quad and \quad}  S_{k}=X_{k}\cap \V(\xb-\ab_0)%_{1}-a_{0,1},\ldots, x_{n}-a_{0,n})
%.$$
Then, the following statements hold for any $k \geq 0$:    \begin{enumerate}[label=(\alph*)]
        \item[$(a)$] $S_k=\V(G_0(\ab_0, \yb),\ldots, G_k(\ab_0, \yb), \xb - \ab_0,z-1).
        $\\[-1.3em]
        \item[$(b)$] 
    $S_{(F_{m},X)}\cap \V(\xb-\ab_0,z-1)\subset S_{k}$.  
    \end{enumerate}
\end{lemma}
\begin{proof}
 $(a)$ For $j\geq 0$, we have
   $\textstyle{F_m^{j}(\xb,\yb,z)=\left(F^{j}(\xb),\;\yb,\; z\cdot\prod_{i=1}^{j-1} h(F^{i}(\xb))\right)}$.
 Then, according to Lemma~\ref{lem:equtionspreimage}, we obtain
 \begin{align*}
     F_m^{-k}(X)\;=\;\V\left(z\cdot \prod_{j=0}^{k-1} h(F^{j}(\xb))
     \cdot \displaystyle\sum_{i=1}^{m} y_{i}g_i(F^{k}(\xb))\right)
     \;=\;\V(z\cdot G_k).
 \end{align*}
Therefore, $
   % \begin{array}{ll}
   X_k =   \V\left(zG_0,\ldots,zG_k\right),\vspace*{0.7em}
   % \end{array}
   $
   %\todoremi{over-technical}
  so we have that  
  \vspace{-4mm}
  % is the intersection for all $0\leq j \leq k$ of the sets
 \begin{align*}
     S_k &= \V(zG_0,\ldots, zG_k, \xb - \ab_0,z-1)\\
     &= \V(G_{0}(\ab_0,\yb),\ldots,G_{k}(\ab_0, \yb) ,\;\xb - \ab_0,z-1).
 \end{align*}

\noindent $(b)$ By Proposition~\ref{prop:stabilization}, we have the following descending chain:
   $$X_{0}\supset X_{1}\supset \ldots \supset X_{N}= S_{(F_{m},X)}=X_{N+1}, \text{ for some } N\in \N.$$
   Thus, by intersecting with $V=\V(\xb -\ab_0,z-1)$ we obtain
    $$S_{0}\supset S_{1}\supset \ldots \supset S_{N}= S_{(F_{m},X)}\cap V %V(x_{1}-a_{0,1},\ldots, x_{n}a_{0,n})
    =S_{N+1}, \text{ for some } N\in \N.$$
   % for some $N\in \N$.
    Thus, $S_{(F_{m},X)}\cap V$ %V(x_{1}-a_{0,1},\ldots, x_{n}-a_{0,n})$ 
    is a subset of $S_{k}$ for any $k\in \N$.
\end{proof}

\begin{remark}
    Given the initial value of a loop, Proposition~\ref{prop_sufficient} can provide as many linear constraints as needed, ensuring that a polynomial invariant in a fixed vector space satisfies these conditions. Since the codimension of $I_{\mL, E}$ is bounded by $m$, the number of generators $g_i$, this bound naturally suggests the number $K$ of linear equations. This leads to a vector subspace $F \subset E$, typically of much lower dimension, that contains $I_{\mL,E}$ and on which the previous algorithms can be executed. As the input size is significantly smaller, this results in a substantial reduction in running time.
\end{remark}
The strategy outlined in the remark above is implemented in Algorithm~\ref{alg2}. We utilize a procedure, \textsf{VectorBasis}, which takes linear forms or a matrix as input and computes a vector basis for the common vanishing set of these forms. This is a standard linear algebra subroutine.

\begin{algorithm}[h]
\setstretch{1.2}
\caption{\TruncatedIdeal}\label{alg2}
\begin{algorithmic}[1]
\Require Polynomials $\gb= (g_1,\ldots,g_m),\hb=(h_1,\ldots, h_k)$ and $F=(f_1,\ldots,f_n)$ in $\mathbb{Q}[\xb]$ and ${\bf a} = (a_1,\ldots,a_n)\in\Q^n$.\vspace*{.3em}
\Ensure A vector basis for $I_{\mL(\ab, \hb, F),E}$ where $E$ is spanned by $g_1,\ldots, g_m$. \vspace*{.3em} % with initial value $\bf a$ and a polynomial map $F$. $\mathcal{L}(\ab,F)$?}
\State\label{step:1} $ g\gets z\cdot(y_1g_1+\cdots+y_mg_m);$
\State $h \gets h_1\cdot \ldots \cdot h_k;$
\State\label{TrunIdStep3} $F_m \gets (F,\; \yb,\; z\cdot h)$
\State\label{TrunIdStep4} $(\bb^1,\ldots,\bb^s)\gets \VectorBasis\Big(g({\bf a,y},1),g(F_m(\ab,\yb,1)), \ldots g(F_m^{m}(\ab, \yb,1)\Big)$;
\State\label{step:B}$\mathcal{B}\gets \Big(\sum_{i=1}^m b_{i}^1g_{i}(\xb),\ldots ,\sum_{i=1}^{m} b_{i}^sg_i(\xb)\Big)$;
\State\label{step:6}$\mathcal{B}'=(P_1,\ldots,P_l) \gets \{ P\in \mathcal{B}\mid \CheckPI({\bf a},P,\hb,F)==\texttt{False}\};$
\If{$\mathcal{B}' ==\emptyset$}\vspace*{-.2em}
%\algstore{testcont}
%\end{algorithmic}
%\end{algorithm}  
%\begin{algorithm}
%\begin{algorithmic}[1]
%\algrestore{testcont}
\State \label{step:D}\Return $\mathcal{B};$\vspace{-.2em}
\EndIf
\State \label{step:C}$A \gets \ComputeMatrix(\mathcal{B}',\hb, F);$
\State$(\cb^{1},\ldots, \cb^{t})\gets \VectorBasis(A(\ab))$;
\State\label{step:12}$\mathcal{C}\gets \big(\sum_{j=1}^{l}c_{i}^1P_j, \ldots, \sum_{j=1}^{l}c_{i}^tP_j\big)$;
\State\label{step:13}$\mathcal{B}''\gets \mathcal{B}.\texttt{remove}(\mathcal{B}'); $
\State \Return $\mathcal{C}.\texttt{extend}(\mathcal{B}'');$
\end{algorithmic}
\end{algorithm}

%\todoremi[inline]{I don't understand the else block of Algorithm 2}

%\begin{remark}Algorithm~\ref{alg2} can be readily extended to compute all polynomial invariants of a specified form, particularly when the potential non-zero terms are predefined, as \textcolro{red}{illustrated in Example~\ref{ex12}.}\end{remark}

\medskip
We now prove the correctness of Algorithm~\ref{alg2}. 
\begin{theorem}\label{thm:corralgo2}
Let $\gb = (g_1, \ldots, g_m)$, $\hb = (h_1, \ldots, h_k)$, and $F = (f_1, \ldots, f_n)$ be polynomials in $\mathbb{Q}[\xb]$. Let $E$ be the vector space spanned by $\gb$. On input a sequence ${\bf a} = (a_1, \ldots, a_n)$ in $\mathbb{Q}^n$, along with $\gb$, $\hb$, and $F$, Algorithm~\ref{alg2} outputs a sequence of polynomials that forms a vector basis for $I_{\mL(\ab,\hb,F),E}$.
\end{theorem}
\begin{proof}
Since $F_m(\xb, \yb, z) = (F(\xb), \yb, zh)$, the polynomials
\[
g(\ab, \yb, 1), \; g(F_m(\ab, \yb, 1)), \dots, \; g(F_m^m(\ab, \yb, 1))
\]
are linear elements of $\mathbb{Q}[y_1, \dots, y_m]$. Let $(\bb^1, \dots, \bb^s)$ be a vector basis of their common vanishing set. According to Proposition~\ref{prop_sufficient}, if $\bb \in \mathbb{C}^n$ satisfies the condition that $g = b_1 g_1 + \dots + b_m g_m$ is a polynomial invariant for $\mathcal{L}(\ab, \hb, F)$, then $\bb$ is a linear combination of the vectors $\bb^i$'s.
Therefore, $I_{\mathcal{L},E}$ is contained in the vector space spanned by the linearly independent polynomials of $E$ in
\vspace{-5mm}
\[
\mathcal{B} = \left( \sum_{i=1}^m b_{1,i} g_i(\xb), \dots, \sum_{i=1}^m b_{s,i} g_i(\xb) \right).
\]
Let $\mathcal{B}' = (P_1, \dots, P_l)$, as defined in Step~\ref{step:6}. By Theorem~\ref{thm:verify}, these are precisely the polynomials in $\mathcal{B}$ that are not polynomial invariants of $\mathcal{L}(\ab, \hb, F)$. 

If $\mathcal{B}'$ is empty, then $\mathcal{B}$ forms a vector basis for $I_{\mathcal{L},E}$, and the proof is complete.
Otherwise, by Theorem~\ref{algogeneral}, the polynomials in $\mathcal{C}$, computed in Step~\ref{step:12}, form a basis for the intersection of $I_{\mathcal{L},E}$ with the vector space spanned by $\mathcal{B}'$. Let $\mathcal{B}'' = \{ g_1, \dots, g_{m-l} \}$, as defined in Step~\ref{step:13}.

We now proceed to prove that $\mathcal{C} \cup \mathcal{B}''$ forms a vector basis of $I_{\mathcal{L},E}$. Since, by construction, these are linearly independent polynomials in $I_{\mathcal{L},E}$, it remains to prove that they span the entire space.
Let $g \in I_{\mathcal{L},E}$. By the argument above, 
$\textstyle{g = \sum_{i=1}^{l} c_i P_i + \sum_{j=1}^{m-l} c_{l+j} g_j}$ for some scalars $c_1, \ldots, c_m \in \mathbb{C}$.
Since $g, g_{l+1}, \ldots, g_m$ all lie in $I_{\mathcal{L},E}$, it follows that $\sum_{i=1}^{l} c_i P_i$ is also in $I_{\mathcal{L},E}$. Thus, this sum lies in the intersection of $I_{\mathcal{L},E}$ with the vector space spanned by $\mathcal{B}'$, and by construction, it is a linear combination of elements in $\mathcal{C}$.
In conclusion, any $g \in I_{\mathcal{L},E}$ can be expressed as a linear combination of elements from $\mathcal{C} \cup \mathcal{B}''$. Thus, the proof is complete.
\end{proof}  
%  \textcolor{blue}{Additionally, it is possible to generate general polynomial invariants of the form $f({\bf x})-f({\bf x_0})=0$ using Algorithm~\ref{algo1} and polynomial invariants that are computed by Algorithm~\ref{alg2} (see Examples~\ref{ex5},~\ref{ex6},~\ref{ex7},~\ref{ex8},~\ref{ex9}).}
%\noindent We now use Algorithm~\ref{alg2} for computing polynomial invariants in an example.% in Example~\ref{ex3}. % up to degree 4. % using Algorithm~\ref{alg2}. 

We now examine the complexity of Algorithm~\ref{alg2}. Building on the approach outlined above, we include data intrinsic to the method (namely, $N$ and $l$) to provide a more precise description of the algorithm's evolution. Notably, the worst-case bounds on $N$ are coarse (see Remark~\ref{remark1}), while the worst-case scenario $l = m$ masks the algorithm's favorable behavior in more advantageous cases.

%Let us now get some insight into the complexity of Algorithm~\ref{alg2}. Following what has been done above, we choose to include data intrinsic to the method (namely $N$ and $l$) in order to describe the evolution of the algorithm more precisely. In fact, the worst-case bounds on $N$ are either coarse (see Remark~\ref{remark1}), and the worst-case $l=m$ hides the nice behavior in more favorable cases. 
\begin{proposition}
Let $\dg$, $\degh$, and $\dF$ denote bounds on the degrees of $\gb$, $\hb$, and $F$, respectively. On input $\gb$, $\hb$, and $F$, Algorithm~\ref{alg2} performs a number of operations in $\Q$ bounded by  
\[
    m \cdot \left(\max\{\dF, k\degh\}^{N_1} \cdot \dg\right)^{O(n^2)}  
    + \left(\max\{\dF, k\degh\}^{N_2} \cdot \dg\right)^{O(n^2 + l^2)},
\]  
where $l$ is the number of polynomials $(P_1, \dotsc, P_l)$ output in Step~\ref{step:6} of Algorithm~\ref{alg2}, and $N_1$ and $N_2$ are the numbers of loop iterations performed in the call to \InvariantSet within the calls to \CheckPI and \ComputeMatrix.  
\end{proposition}
\begin{proof}
The complexity of Steps~\ref{step:1} through~\ref{step:B} involves reasonable polynomial multiplications, iterative evaluations, and linear algebra in dimension at most $m$. Consequently, these operations are all bounded by the above complexity estimate.  
Let $s$ denote the number of polynomials $(P_1, \dotsc, P_s)$ computed in Step~\ref{step:B}. Step~\ref{step:6} then consists of at most $s \leq m$ calls to \CheckPI, whose total complexity is bounded by  
    $\textstyle{m \cdot \left(\max\{\dF, k\degh\}^{N_1} \cdot \dg\right)^{O(n^2)}}$,
applying Theorem~\ref{thm:verify}, since $\deg(P_i) \leq \dg$ for all $1 \leq i \leq s$. This accounts for the first part of the summand.  
The second part arises solely from Step~\ref{step:C}, where the complexity estimate follows directly from Theorem~\ref{algogeneral}, as only $l$ generators are considered in this step.  
\end{proof}

\begin{remark}\label{remark_disequality_fixed}
%As described in Remark~\ref{remark_disequality}, for the general case, one can easily adapt this algorithm to handle loops with disequalities $p(x)\neq 0$ in the guard by replacing $g$ with $p\cdot g$ at step 3.
In the worst-case scenario where $l = m$, Algorithm~\ref{alg2} shares the same complexity bounds as Algorithm~\ref{matrixalgo}, since the latter is called on the same input as the former, apart from the absence of an initial value.  
%Specifically, if no candidates are found at Step~\ref{step:B}, Algorithm~\ref{matrixalgo} would exhibit the same behavior at Step~\ref{step:C}, without utilizing the given initial values. 
In practice, however, all candidates identified in Step~\ref{step:B} are invariants (see Section~\ref{sec:implementation}), meaning $l = 0$. As a result, Algorithm~\ref{alg2} terminates at Step~\ref{step:D}, with the total cost bounded by $m$ invariant checks.  
\end{remark}

\begin{example}[Squares]\label{ex3}
Consider the ``Squares'' loop:

\programboxappendix[0.55\linewidth]{
\State $(x_1,x_2, x_3)=(-1,-1, 1)$
\While{true}
\State $\begin{pmatrix}
x_1 \\
x_2\\
x_3
\end{pmatrix}
\longleftarrow
%\xleftarrow{\textbf{F}}
\begin{pmatrix}
2x_1+x_2^2+x_3\\
2x_2-x_2^2+2x_3\\
1-x_3
\end{pmatrix}
$

\EndWhile
}

\noindent For $d\geq 2$, Algorithm~\ref{matrixalgo}
 cannot compute a polynomial matrix $A$ such that 
\vspace{-5mm} \[
 I_{\mL(\ab,1,F),\Q[\xb]_{\leq d}}=\left\{\sum_{|\alpha_i|\leq d}b_i \xb^{\alpha_i}\mid \bb\in \ker A(\ab)\right\},
 \]
for any $\ab \in \C^{3}$, within an hour.
However, when the initial values are fixed, Algorithm~\ref{alg2} computes all polynomial invariants up to degree 5 within 2 seconds. To compute $I_{\mathcal{L}, \mathbb{Q}[\xb]_{\leq 2}}$, we call Algorithm~\ref{alg2} with input $((-1,-1,1), \gb, 1, F)$, where $\gb$ is the set of all monomials up to degree 2. Assume that
$g = b_1 + b_2 x_1 + \cdots + b_{10} x_3^2 $
is a polynomial invariant.
 % of degree $2$. 
By Proposition~\ref{prop_sufficient}, this leads to 10 linear equations, whose solutions 
give the following 5 candidates in $\mathcal{B}$ %for polynomial invariants: 
\[
  \begin{array}{lll}
  \{1+x_1+x_2+x_3,\;\; 1+x_1+x_2+x_3^2,\;\; 
  2+3(x_1+x_2) + (x_1+x_2)^2,\\[.5em]
 x_1^2-x_2^2+2x_1x_3-x_1-3x_2-2,\;\; x_2^2-x_1^2+2x_2x_3-x_2-3x_1+2\}.
  \end{array}
  \]
The algorithm $\CheckPI$ then verifies that all polynomials in $\mathcal{B}$ are invariant and that $\mathcal{B}$ forms a basis for $I_{\mathcal{L}, \mathbb{Q}[\xb]_{\leq 2}}$, which has dimension 5.
Additional computations show that $I_{\mathcal{L}, \mathbb{Q}[\xb]_{\leq 3}}$ and $I_{\mathcal{L}, \mathbb{Q}[\xb]_{\leq 4}}$ have dimensions 13 and 26.
\end{example}
Example~\ref{ex3} was previously explored in \cite{Unsolvableloops}, where the closed formula
\[
x_1(n) + x_2(n) = 2^n(x_1(0) + x_2(0) + 2) - \frac{(-1)^n}{2} - \frac{3}{2}
\]
was derived, with $x_i(n)$ representing the value of $x_i$ after $n$ iterations of the loop. Notably, none of the invariants listed above were found. 
Using Algorithm~\ref{alg2}, we compute the ideals $I_{\mathcal{L}, \mathbb{Q}[\xb]_{\leq d}}$ for $d = 1, 2, 3, 4$, given a specific initial value. This captures all polynomial invariants up to degree 4.
\begin{example}\label{ex: ps6}%[ps6\footnote{\url{https://github.com/sosy-lab/sv-benchmarks/blob/master/c/nla-digbench}}]. 
Consider the ``ps6" loop\footnote{\url{https://github.com/sosy-lab/sv-benchmarks/blob/master/c/nla-digbench}} $\mL$:

\programboxappendix[0.5\linewidth]{
\State$(x_1, x_2)=(0,0)$
\While{$x_2-18665\neq 0$}
\State $\begin{pmatrix}
x_1 \\
x_2
\end{pmatrix}
\longleftarrow
\begin{pmatrix}
x_1+x_2^5\\
x_2+1
\end{pmatrix}
$
\EndWhile
}

\noindent 
This loop computes the sum of the fifth powers of the first $n$ natural numbers after $n$ iterations. Using Algorithm~\ref{alg2}, we find the formula for this sum. Specifically, the algorithm reveals that the only polynomial invariant of degree at most 6 (up to scalar multiplication) is
$$ x_1 - \left( \frac{1}{6} x_2^6 - \frac{1}{2} x_2^5 + \frac{5}{12} x_2^4 - \frac{1}{12} x_2^2 \right). $$
After $n+1$ iterations of $\mathcal{L}$, the value of $x_1$ is $1^5 + 2^5 + \cdots + n^5$ and $x_2$ is $n+1$. From this invariant, we deduce the formula for the sum of fifth powers:
\[
1^5 + 2^5 + \cdots + n^5 = \frac{1}{6}(n+1)^6 - \frac{1}{2}(n+1)^5 + \frac{5}{12}(n+1)^4 - \frac{1}{12}(n+1)^2.
\]
\end{example}

\subsection{Generalization to branching loops}\label{subsection:branching}
%%%%%%%%%%%%%%%%%%%%%%%%%%%%%%%%%%%%%%%%%%%%%%%%%%%%%%%%%%
%%%%%%%%%%%%%%%%%%%%%%%%%%%%%%%%%%%%%%%%%%%%%%%%%%%%%%%%%%
In this subsection, we present a method to generate all polynomial invariants up to a specified degree for branching loops with a nondeterministic conditional statement involving $k$ branches. Branching loops generalize the types of loops discussed in earlier sections, as they account for multiple scenarios through distinct maps. %This subsection extends the definitions and results introduced previously. 

\begin{definition}
   Let $\ab = (a_1, \ldots, a_n) \in \mathbb{C}^n$, $h\in \mathbb{C}[\xb]$, and let $F_1, \ldots, F_k : \mathbb{C}^n \to \mathbb{C}^n$ be polynomial maps. Then, $\mathcal{L}(\ab, h, (F_1, \ldots, F_k))$ represents a branching loop with a nondeterministic conditional statement involving $k$ branches, as follows:
    \vspace{3mm}
    
    \programbox[0.6\linewidth]{
\State$(x_{1},\ldots, x_n)=(a_1,\ldots, a_n)$
\While{$h\neq  0$}
\If{$\ast$}
\State$(x_{1},\ldots, x_n)\gets F_{1}(x_{1},\ldots, x_n)$\\
\hspace{5mm}$\textbf{\ldots}$
\ElsIf{$\ast$}
\State$(x_{1},\ldots, x_n)\gets F_{i}(x_{1},\ldots, x_n)$\\
\hspace{5mm}$\textbf{\ldots}$
\Else
\State$(x_{1},\ldots, x_n)\gets F_{k}(x_{1},\ldots, x_n)$
\EndIf
\EndWhile
}

\noindent When no $h$ is defined, we will write $\mathcal{L}(\ab,1,(F_1,\ldots, F_k))$.
\end{definition}
\begin{definition}
Let $F_1, \ldots, F_k : \mathbb{C}^n \to \mathbb{C}^n$ be polynomial maps. For any $m \in \mathbb{N}$ and any sequence $i_1, \ldots, i_m \in [k]$, define the polynomial map  
$$F_{i_1, \ldots, i_m}(\xb) = F_{i_m} \big(F_{i_{m-1}}(\cdots F_{i_1}(\xb) )\big).$$  \end{definition}
Polynomial invariants for branching loops without guard conditions were studied in~\cite{rodriguez2007generating}. Here, we extend this definition to include cases where the guard involves an inequation. We denote $[k] = \{1, \ldots, k\}$. 
\begin{definition}
A~polynomial~$g$ is an invariant of the branching loop $\mathcal{L}(\ab, h, (F_1, \ldots, F_k))$~if,~for~any~$m \in \mathbb{N}$~and any sequence~$i_1, \ldots, i_m \in [k]$, either:  
$$g(F_{i_1}(\ab)) = \cdots = g(F_{i_1, \ldots, i_m}(\ab)) = 0,$$  
or there exists $l < m$ such that  
\[
g(F_{i_1}(\ab)) = \cdots = g(F_{i_1, \ldots, i_l}(\ab)) = 0 \quad \text{and} \quad h(F_{i_1, \ldots, i_l}(\ab)) = 0.
\]
The set of all polynomial invariants for $\mathcal{L}(\ab, h, (F_1, \ldots, F_k))$, denoted by $I_{\mathcal{L}(\ab, h, (F_1, \ldots, F_k))}$, is called the \textit{invariant ideal} of $\mathcal{L}(\ab, h, (F_1, \ldots, F_k))$.  
\end{definition}
In~\cite{hrushovski2023strongest}, the authors showed that the invariant ideals of nondeterministic branching loops are computable when the associated polynomial maps are linear. However, they also established that in the general case, these invariant ideals are not computable. For extended P-solvable nondeterministic branching loops, the computation of invariant ideals is addressed in~\cite{humenberger2017invariant}. The generation of polynomial inequality invariants for nondeterministic branching loops is further explored in~\cite{chatterjee2020polynomial}.

We now extend the concept of invariant sets to encompass multiple polynomial maps, enabling their application in generating polynomial invariants for branching loops.

\begin{definition}\label{def:invariantset}
Let $F_1,\ldots, F_k : \Cn \longrightarrow \Cn$ be maps and $X$ be a subset of $\Cn$.
The invariant set $S_{((F_1,\ldots, F_k),X)}$ of $((F_1,\ldots, F_k),X)$ is defined as:
\begin{center}
    $\Big\{x \in X \mid \forall m\in \N, \forall  i_1,\ldots, \forall i_m\in [k],\; F_{i_m}\big(\ldots(F_{i_1}(x)\big) \in X \Big\}.$
\end{center}
\end{definition}
We present a method for computing invariant sets of multiple polynomial maps using an iterative, converging outer approximation, analogous to the approach outlined in Proposition~\ref{prop:stabilization} for a single polynomial map. Building on the following proposition, we further extend Algorithm~\ref{algo1} to handle the case of multiple polynomial maps.
\begin{proposition}\label{prop: extendinv}
Let $F_1, \ldots, F_k : \mathbb{C}^n \to \mathbb{C}^n$ be polynomial maps and let $X \subset \mathbb{C}^n$ be an algebraic variety. Define $X_0 = X$, and for $i \geq 1$, let  
$$ X_{i+1} = X_i \cap F_1^{-1}(X_i) \cap \ldots \cap F_k^{-1}(X_i). $$  
Then, there exists an integer $N$ such that $X_N = X_{N+1}$, and the invariant set $S_{((F_1, \ldots, F_k), X)}$ is precisely equal to $X_N$.  \end{proposition}
\begin{proof}
Following the strategy in Proposition~\ref{prop:stabilization}, one can show that there exists $N \in \mathbb{N}$ such that $X_N = X_{N+1}$, and for this $N$, $X_N = X_m$ for all $m \geq N$. It remains to prove that this $X_N$ is the invariant set defined above.

First, we prove by induction on $m$ that $x \in X_m$ if and only if $F_{i_1, \ldots, i_s}(\mathbf{x}) \in X$ for any $s \leq m$ and any $i_1, \ldots, i_s \in [k]$.  
By definition, $X_0 = X$, which establishes the base case.
Now, assume that for some $m > 0$, $x \in X_{m-1}$ if and only if $F_{i_1, \ldots, i_s}(\mathbf{x}) \in X$ for any $s \leq m-1$ and any $i_1, \ldots, i_s \in [k]$.  
If $\mathbf{x} \in X_m$, then  
\[
\mathbf{x} \in X_{m-1} \cap F_1^{-1}(X_{m-1}) \cap \ldots \cap F_k^{-1}(X_{m-1}).
\]
Thus, for any $s \leq m-1$ and any $i_1, \ldots, i_s \in [k]$, we have  
$F_{i_1, \ldots, i_s}(\mathbf{x}) \in X,$  
and for any $i_1, \ldots, i_s, i \in [k]$,  
$F_{i_1, \ldots, i_s, i}(\mathbf{x}) \in X.$  
Finally, for any $s \leq m$ and any $i_1, \ldots, i_s \in [k]$,  
$F_{i_1, \ldots, i_s}(\mathbf{x}) \in X.$
The converse holds in a similar manner.

%To prove the converse, it suffices to reverse the argument.

We can now conclude the proof of the proposition. Since $X_N = X_m$ for all $m \geq N$, it follows that $X_N$ is the set of all $\mathbf{x} \in X$ such that  
$F_{i_1, \ldots, i_m}(\mathbf{x}) \in X$  
for every $m \in \mathbb{N}$ and $i_1, \ldots, i_m \in [k]$. This is precisely the invariant set as defined in Definition~\ref{def:invariantset}.
\end{proof}

\begin{theorem}\label{thm:InvBranch}
  Let $F_1, \ldots, F_k : \mathbb{C}^n \to \mathbb{C}^n$ be polynomial maps, $F=(F_1,\ldots,F_k)$ and let $\gb = (g_1, \ldots, g_m)$ be a collection of polynomials in $\mathbb{Q}[\mathbf{x}]$. There exists an algorithm, \InvariantSetBranch, which takes as input $(\gb, F)$ and outputs a sequence of polynomials whose common vanishing set corresponds to the invariant set $S_{(F, \V(\gb))}$.
\end{theorem}
\begin{proof}
Let \InvariantSetBranch be the modified version of Algorithm~\ref{algo1}, obtained by making the following replacements:  
\begin{itemize}  
    \item %Replace 
    $\compose(\gb, F)$ with $(\compose(\gb, F_1), \ldots, \compose(\gb, F_k))$ in Step~\ref{InSStep2}.  
    \item %Replace 
    $\compose(\widetilde{\gb}, F)$ with $(\compose(\widetilde{\gb}, F_1), \ldots, \compose(\widetilde{\gb}, F_k))$ in Step~\ref{InSStep5}.  
\end{itemize}  
Let $S_0 = \gb_0$, and for $m \geq 1$, define $S_m$ as the set contained in $S$ after $m$ iterations of the \textbf{while} loop. Let $\widetilde{\gb}_0 = \gb$, $\widetilde{\gb}_1 = (\gb \circ F_1, \ldots, \gb \circ F_k)$, and for $m \geq 1$, let $\widetilde{\gb}_{m+1}$ denote the sequence contained in $\widetilde{\gb}$ after $m$ iterations.  

Let $X_0 = \V(\gb)$, and for $m \geq 1$, let $X_m$ be as defined in Proposition~\ref{prop: extendinv}. Then, as shown in the proof of Theorem~\ref{thm:corralgo1}, it follows that $X_m = \V(S_m)$.  

Moreover, following the proof of Theorem~\ref{thm:corralgo1} and by Proposition~\ref{prop: extendinv}, there exists an integer $N$ such that the loop in \InvariantSetBranch terminates after $N$ iterations. At this point, the algorithm outputs polynomials in $S_N$, ensuring that $S_{((F_1, \ldots, F_k), \V(\gb))} = \V(S_N)$.
\end{proof}
We now present a necessary and sufficient condition for identifying polynomial invariants of branching loops through the computation of their invariant sets, enabling us to adapt Algorithm~\ref{algo:verify} to handle branching loops.

\begin{proposition}\label{prop:branchverify}
 Let $F_1, \ldots, F_k: \mathbb{C}^n \longrightarrow \mathbb{C}^n$ be polynomial maps, $F=(F_{1},\ldots,F_{k})$, and $h$ and $g$ be polynomials in $\mathbb{C}[\xb]$. Let $z$ be a new indeterminate, $X = \V(zg) \subset \mathbb{C}^{n+1}$, and for any $i \in [k]$, define  
$$ F_{i,0}(\mathbf{x}, z) = (F_i(\mathbf{x}), z h(\mathbf{x}))$$  
and let $\Fb_0=(F_{1,0},\ldots,F_{k,0})$. Then, for any $\mathbf{a} \in \mathbb{C}^n$, % we have
$g(\mathbf{x}) \in I_{\mathcal{L}(\mathbf{a}, h, (F_1, \ldots, F_k))}$ if and only if $(\mathbf{a}, 1) \in S_{(\Fb_0, X)}$.
\end{proposition}
\begin{proof}
Let $\ab_0=\ab$ and for $l\geq 1$, let $\ab_{i_1,\ldots,i_m}=F_{i_m}(\ldots (F_{i_1}(\ab_0)))$. Denote %$F=(F_1,\ldots, F_k)$
%$, \Fb_0=(F_{1,0},\ldots,F_{k,0})$ 
%and 
$F_{i_1,\ldots, i_m,0}(\xb)=F_{i_m,0}(\ldots (F_{i_1,0}(\xb)))$.
Assume $g(\xb)$ is a polynomial invariant of $ \mL(\ab, h,F)$. Thus, for any $m\in \N$ and any $i_1,\ldots,  i_m \in [k]$, either $g(\ab_{i_1})=\cdots =g(\ab_{i_1,\ldots, i_m})=0$
or there exists $l< m$ such that
$$g(\ab_{i_1})=\cdots =g(\ab_{i_1,\ldots, i_l})=0 \text{ and } h(\ab_{i_1,\ldots, i_l})=0.
$$
From the definition of polynomial maps $F_{i,0}$, we obtain
$$F_{i_1,\ldots, i_m,0}(\ab,1)=(\ab_{i_1,\ldots, i_m},h(\ab_0)\cdot h(\ab_{i_1})\cdot \ldots\cdot h(\ab_{i_1,\ldots, i_{m-1}})).$$
Therefore, $(zg)\circ F_{i_1,\ldots, i_m, 0}(\ab,1)=0$. Hence, $F_{i_1,\ldots, i_m,0}(\ab,1)\in X$, for any $m\in N$ and any $i_1,\ldots, i_m\in [k]$, that is $(\ab,1)\in S_{(\Fb_{0},X)}.$

Conversely, assume that $(\mathbf{a}, 1) \in S_{(\Fb_0, X)}$. Let $m \in \mathbb{N}$ and $i_1, \ldots, i_m \in [k]$. For any $l \leq m$, $F_{i_1, \ldots, i_l, 0}(\mathbf{a}, 1)$ is in $X$, implying that  
$$ h(\mathbf{a}_0) \cdot h(\mathbf{a}_{i_1}) \cdot \ldots \cdot h(\mathbf{a}_{i_1, \ldots, i_{l-1}}) g(\mathbf{a}_{i_1, \ldots, i_l}) = 0. $$  
It follows that either  
$g(\mathbf{a}_{i_1}) = \cdots = g(\mathbf{a}_{i_1, \ldots, i_m}) = 0,$  
or for some $l < m$,  
$$ g(\mathbf{a}_{i_1}) = \cdots = g(\mathbf{a}_{i_1, \ldots, i_l}) = 0 \text{ and } h(\mathbf{a}_{i_1, \ldots, i_l}) = 0, $$  
which implies that $g(\mathbf{x}) = 0$ is a polynomial invariant of $\mathcal{L}(\mathbf{a}, h, F)$.  
\end{proof}

    \begin{theorem}\label{thm:CheckBranch}
   Let $F_1, \ldots, F_k : \mathbb{C}^n \to \mathbb{C}^n$ be polynomial maps, $F=(F_1,\ldots,F_k)$, $\mathbf{a} \in \mathbb{Q}^n$, and let $g, \mathbf{h} = (h_1, \ldots, h_l)$ be polynomials in $\mathbb{Q}[\mathbf{x}]$. There exists an algorithm \CheckPIBranch that takes the input $(\mathbf{a}, g, \mathbf{h}, F)$ and outputs {\normalfont\texttt{True}} if $g \in I_{\mathcal{L}(\mathbf{a}, \mathbf{h}, F)}$, and {\normalfont\texttt{False}} otherwise.
\end{theorem}
\begin{proof}
 Let \CheckPIBranch be the algorithm derived from Algorithm~\ref{algo:verify} by making the following modifications:
\begin{itemize}
    \item replacing $(F, z h)$ with $((F_1, z h), \ldots, (F_k, z h))$ in Step~\ref{CheckPIStep2};
    \item replacing \InvariantSet with \InvariantSetBranch in Step~\ref{CheckPIStep3}.
\end{itemize}

Let $X = \V(z g) \subset \mathbb{C}^{n+1}$. By Theorem~\ref{thm:InvBranch}, the invariant set $S_{(\Fb_0, X)}$ corresponds to the vanishing set of $P_1, \ldots, P_m$. Therefore, according to Proposition~\ref{prop:branchverify}, $g \in I_{\mathcal{L}(\mathbf{a}, \mathbf{h}, F)}$ if and only if $P_1(\mathbf{a}, 1) = \ldots = P_m(\mathbf{a}, 1) = 0$, that is, if and only if \CheckPIBranch returns \texttt{True}.
\end{proof}

We now present a criterion for identifying polynomial invariants of branching loops within a vector subspace $E$ of $\mathbb{C}[\mathbf{x}]$. This enables us to extend Algorithm~\ref{matrixalgo} to handle the case of branching loops.

\begin{proposition}\label{prop:branching}
 Let $F_1, \ldots, F_k: \mathbb{C}^n \longrightarrow \mathbb{C}^n$ be polynomial maps, $F=(F_1,\ldots,F_k)$, and let $h, g_1, \ldots, g_m$ be polynomials in $\mathbb{C}[\xb]$.~Let $E$ be the vector space spanned by $(g_1, \ldots, g_m)$.~Let $\mathbf{y} = (y_1, \ldots, y_m)$ be new indeterminates.~Define  
$$ g(\mathbf{x}, \mathbf{y}) = y_1 g_1(\mathbf{x}) + \cdots + y_m g_m(\mathbf{x}).$$  
Let $z$ be a new indeterminate, and $X = \V(zg) \subset \mathbb{C}^{n+m+1}$. Define  
$$ F_{i,m}(\mathbf{x}, \mathbf{y}, z) = (F_i(\mathbf{x}), \mathbf{y}, z h(\mathbf{x})),$$  
for any $i \in [k]$. Then, for any $\mathbf{a} \in \mathbb{C}^n$,  
\[
I_{\mathcal{L}(\mathbf{a}, h, F)} = \{ g(\mathbf{x}, \mathbf{b}) \mid (\mathbf{a}, \mathbf{b}, 1) \in S_{((F_{1,m}, \ldots, F_{k,m}), X)} \}.
\]
\end{proposition}
\begin{proof}
     Define $G_i(\xb,\yb)=(F_i(\xb),\yb)$. Since $g(G_{i_1,\ldots, i_l}(\ab,\bb))=g(F_{i_1,\ldots, i_l}(\ab),\bb)$,
    $$g(\xb,\bb)\in I_{\mL(\ab, h, F)} \iff g(\xb,\yb)\in I_{\mL((\ab,\bb), h,(G_1,\ldots, G_k))}.$$
By Proposition~\ref{prop:branchverify}, $g(\xb,\yb)\in I_{\mL((\ab,\bb), h,(G_1,\ldots, G_k))}$ if and only if $$(\ab,\bb,1)\in S_{((G_{1,0},\ldots, G_{k,0}),X)}.$$ Since $G_{i,0}$ is the same polynomial map as $F_{i,m}$ for any $i$, we have
$$g(\xb,\bb)\in I_{\mL(\ab, h, F)} \iff (\ab,\bb,1)\in S_{((F_{1,m},\ldots, F_{k,m}),X)}.$$
which proves the proposition.
\end{proof}

%\begin{erdenebayar}
    \begin{theorem}\label{thm:ComputeBranch}
  Let $F_1, \ldots, F_k : \mathbb{C}^n \to \mathbb{C}^n$ be polynomial maps, $F = (F_1, \ldots, F_k)$, and $\mathbf{h} = (h_1, \ldots, h_l)$, $\gb = (g_1, \ldots, g_m)$ be polynomials in $\mathbb{Q}[\mathbf{x}]$. Let $E$ be the vector space spanned by $\gb$. Then, there exists an algorithm \ComputeMatrixBranch that takes the input $(\gb, \mathbf{h}, F)$ and outputs a polynomial matrix $A$ such that for any $\mathbf{a} \in \mathbb{C}^n$,
\[
I_{\mathcal{L}(\mathbf{a}, \mathbf{h}, F), E} = 
\left\{ \sum_{i \leq m} b_i g_i \mid (b_1, \ldots, b_m) \in \ker A(\mathbf{a}) \right\}.
\]
\end{theorem}
\begin{proof}
Let \ComputeMatrixBranch be the algorithm derived from Algorithm~\ref{matrixalgo} by making the following modifications:
\begin{itemize}
    \item replacing $(F, \mathbf{y}, z h)$ with $((F_1, \mathbf{y}, z h), \ldots, (F_k, \mathbf{y}, z h))$ in Step~\ref{ComMatStep3};
    \item replacing \InvariantSet with \InvariantSetBranch in Step~\ref{ComMatStep4}.
\end{itemize}
Let $g = y_1 g_1 + \cdots + y_m g_m$ and $h = h_1 \cdots h_l$ as in Algorithm~\ref{matrixalgo}, and define $X = \V(zg)$. % \subset \mathbb{C}^{n+m+1}$. 
By Theorem~\ref{thm:InvBranch}, on input $(zg, ((F_1, \mathbf{y}, z h), \ldots, (F_k, \mathbf{y}, z h)))$, \InvariantSetBranch outputs polynomials $P_1, \ldots, P_N \in \mathbb{Q}[\mathbf{x}, \mathbf{y}, z]$ whose common vanishing set is $S_{((F_{1,m}, \ldots, F_{k,m}), X)}$. Let $\widetilde{P}_j = P_j(\mathbf{x}, \mathbf{y}, 1)$ for each $j$. %$ \in \{1, \ldots, N\}$. 
By the construction of \InvariantSetBranch, for each $j \in \{1, \ldots, N\}$, there exist indices $i_1, \ldots, i_{s_j} \in [k]$ such that
\begin{align*}
    \widetilde{P}_j &= P_j(\mathbf{x}, \mathbf{y}, 1) = (zg) \circ (F_{i_1, \ldots, i_{s_j}}(\mathbf{x}), \mathbf{y}, h(\mathbf{x}) \cdot h(F_{i_1}(\mathbf{x})) \cdot \ldots \cdot h(F_{i_1, \ldots, i_{s_j-1}}(\mathbf{x}))) \\
    &= h(\mathbf{x}) \cdot h(F_{i_1}(\mathbf{x})) \cdot \ldots \cdot h(F_{i_1, \ldots, i_{s_j-1}}(\mathbf{x})) \cdot g(F_{i_1, \ldots, i_{s_j}}(\mathbf{x}), \mathbf{y}).
\end{align*}
Since $g(\mathbf{x}, \mathbf{y})$ is linear in the $\mathbf{y}_i$'s, the $P_j$'s are linear in the $\mathbf{y}_i$'s. By applying Proposition~\ref{prop:branching}, the rest of the proof follows similarly to the non-branching case in Theorem~\ref{algogeneral}.
\end{proof}

We conclude this subsection by extending Algorithm~\ref{alg2}, which computes polynomial invariants for a fixed initial value, to the case of branching loops. First, we generalize Proposition~\ref{prop_sufficient} to this setting, providing a sufficient condition for a polynomial to be an invariant of a branching loop, based on the initial values of the loop.

\begin{proposition}\label{prop_sufficientbranching}
Consider a loop $\mathcal{L}(\mathbf{a}_0, h(\mathbf{x}), (F_1, \ldots, F_k))$. Let $g_1, \ldots, g_m$ %\textcolor{red}{$g_1(\mathbf{x}), \ldots, g_m(\mathbf{x})$ }
be polynomials in $\mathbb{C}[\mathbf{x}]$. Define 
$G_0(\mathbf{x}, \mathbf{y}) = y_1 g_1(\mathbf{x}) + \cdots + y_m g_m(\mathbf{x}),$  
and  
$$ G_{i_1, \ldots, i_s}(\mathbf{x}, \mathbf{y}) = \prod_{j=0}^{s-1} h(F_{i_1, \ldots, i_j}(\mathbf{x})) \cdot \sum_{i=1}^{m} y_i g_i(F_{i_1, \ldots, i_s}(\mathbf{x})),$$
for any $s \geq 1$ and $i_1, \ldots, i_s \in [k]$.
If $\mathbf{b} \in \mathbb{C}^m$ and $G_0(\mathbf{x}, \mathbf{b})$ is a polynomial invariant of $\mathcal{L}(\mathbf{a}_0, h, (F_1, \ldots, F_k))$, then for any $s \geq 1$ and $i_1, \ldots, i_s \in [k]$,  
\begin{equation*}\label{le}
G_0(\mathbf{a}_0, \mathbf{y}) = G_{i_1, \ldots, i_s}(\mathbf{a}_0, \mathbf{y}) = 0.
\end{equation*} 
\end{proposition}
Proposition~\ref{prop_sufficientbranching} is a direct consequence of~the~following~lemma.

\begin{lemma}\label{lemma:branching}
Let $F_1,\ldots ,F_k:\Cn \longrightarrow \Cn$ be polynomial maps and $h, g_1,\ldots g_m$ be polynomials in $\C[\xb]$.
Let $$X_0 = \V(z(y_1g_1(\xb)+\cdots +y_mg_m(\xb)))\subset \mathbb{C}^{n+m+1}$$  
and for $l\geq 1$ define $X_{l} = X_{l-1}\cap\displaystyle\bigcap_{i=1}^{k}F_{i,m}^{-1}(X_{l-1})$. Let $\ab_0\in \Cn$ and $$S_{l}=X_{l}\cap \V(\xb-\ab_0,z-1),$$
 for all $l \in \N$.
Finally, let $\mathbf{G}_{\leq l}$ denote the list containing $G_0$ and all $G_{i_1, \ldots, i_j}$ for $1 \leq j \leq l$ and $i_1, \ldots, i_j \in [k]$.
Then for any $l \in \N$, the following holds:
    \begin{itemize}
    \item[$(a)$] $S_l=\V\Big(\mathbf{G}_{\leq l}(\ab_0,y), \;\xb-\ab_0,\; z-1\Big);$
    \item[$(b)$] 
    $S_{((F_{1,m},\ldots, F_{k,m}),X)}\cap \V(\xb-\ab_0,z-1)\subset S_{l}$.
    \end{itemize}
\end{lemma}
\begin{proof} Let $F_{i_1,\ldots, i_j;m}(\xb)=F_{i_j,m}(\ldots (F_{i_1,m}(\xb)))$. 
 We now prove $$X_l=\displaystyle\bigcap_{j=0}^{l}\bigcap_{i_1,\ldots, i_j\in [k]} F_{i_1,\ldots,i_j;m}^{-1}(X_0)$$ by induction on $l$ and define $F_{\emptyset}(\xb)=\xb$. For $l=0$, we have $X_0=F^{-1}_{\emptyset;m}(X_0)$
 which proves the base case. Assume that the statement holds for $l-1\geq 0$,  that is $\textstyle{X_{l-1}=\bigcap_{j=0}^{l-1}\bigcap_{i_1,\ldots, i_j\in [k]} F_{i_1,\ldots,i_j;m}^{-1}(X)}.$ ~From~the~inductive~hypothesis, it follows that:
  $
  \textstyle\bigcap_{i=0}^{k}F_{i,m}^{-1}(X_{l-1})=\bigcap_{j=1}^{l}\bigcap_{i_1,\ldots, i_j\in [k]} F_{i_1,\ldots,i_j;m}^{-1}(X).
  $
 Since $X_l =  X_{l-1}\cap\displaystyle\bigcap_{i=1}^{k}F_{i,m}^{-1}(X_{l-1})$, the proof is complete.
 
\noindent $(a)$ The case $l = 0$ follows directly from the definitions. Let $l \geq 1$, then
$$ F_{i_1, \ldots, i_j; m}(\mathbf{x}, \mathbf{y}, z) = \left( F_{i_1, \ldots, i_j}(\mathbf{x}), \mathbf{y}, z \prod_{s=0}^{j-1} h(F_{i_1, \ldots, i_s}(\mathbf{x})) \right)$$ 
for $1 \leq j \leq l$. Then, according to Lemma~\ref{lem:equtionspreimage},  
\begin{align*}
    F_{i_1, \ldots, i_j; m}^{-1}(X) &= \V\left( z \prod_{s=0}^{j-1} h(F_{i_1, \ldots, i_s}(\mathbf{x})) \cdot \sum_{t=1}^{m} y_t g_t(F_{i_1, \ldots, i_j}(\mathbf{x})) \right) \\
    &= \V(z \cdot G_{i_1, \dotsc, i_j}).
\end{align*}  
Since $X_l = \bigcap_{j=0}^{l} \bigcap_{i_1, \ldots, i_j \in [k]} F_{i_1, \ldots, i_j; m}^{-1}(X_0)$, we are done. Indeed, the last two equations allow us to substitute $\mathbf{x}$ and $z$ with $\mathbf{a}_0$ and $1$, respectively.

\medskip\noindent The proof of item $(b)$ is similar to the non-branching case in Lemma~\ref{prop4}.
\end{proof}
%\begin{erdenebayar}
\begin{theorem}\label{thm:branchingloops}
Let $F_1,\ldots ,F_k:\Cn \longrightarrow \Cn$ be polynomial maps, $F=(F_1,\ldots,F_k)$, $\ab \in \Q^n$. Consider the polynomials  $\hb=(h_1,\ldots, h_l)$ and $\gb=(g_1,\ldots g_m)$ in $\C[\xb]$. Let $E$ be the vector space spanned by $\gb$. Then, there exists an algorithm \TruncatedIdealBranch that takes the input $(\ab, \gb,\hb, F)$ and outputs a sequence of polynomials that forms a vector basis for $I_{\mL(\ab,\hb,F),E}$. 
\end{theorem}
\begin{proof}
Let \TruncatedIdealBranch be the modified version of Algorithm~\ref{alg2} obtained by replacing: \begin{itemize} \item $F$ with $(F_1, \ldots, F_k)$ in Steps~\ref{TrunIdStep3},~\ref{step:6} and~\ref{step:C}; \item the linear equations in Step~\ref{TrunIdStep4} with those described in Proposition~\ref{prop_sufficientbranching}; \item \CheckPI with \CheckPIBranch in Step~\ref{step:6}; \item \ComputeMatrix with \ComputeMatrixBranch in Step~\ref{step:C}. \end{itemize}
By Proposition~\ref{prop_sufficientbranching}, $I_{\mathcal{L}(\mathbf{a}, \mathbf{h}, F), E}$ is contained in the vector space spanned by the linearly independent polynomials of $E$ in $\mathcal{B}$, as defined in Step~\ref{step:B} of \TruncatedIdealBranch. By Theorem~\ref{thm:CheckBranch}, $\mathcal{B}'$, defined in Step~\ref{step:6} of \TruncatedIdealBranch, consists of all polynomials in $\mathcal{B}$ that are not polynomial invariants of $\mathcal{L}(\mathbf{a}, \mathbf{h}, F)$. By Theorem~\ref{thm:ComputeBranch}, the polynomials in $\mathcal{C}$, computed in Step~\ref{step:12}  of \TruncatedIdealBranch, form a basis for the intersection of $I_{\mathcal{L}(\mathbf{a}, \mathbf{h}, F), E}$ with the vector space spanned by $\mathcal{B}'$. Let $\mathcal{B}''$ be as defined in Step~\ref{step:13} of \TruncatedIdealBranch. Then, $\mathcal{C} \cup \mathcal{B}''$ forms a basis for $I_{\mathcal{L}(\mathbf{a}, \mathbf{h}, F), E}$, and the proof follows in the same way as the non-branching case in Theorem~\ref{thm:corralgo2}.
\end{proof}
\begin{example}[Markov triples]\label{ex: markov}
Consider the following loop $\mL$:

\programbox[0.5\linewidth]{
    \State $(x_1,x_2,x_3) := (1,1,2)$
    \While{true}
        \If{$\ast$}
            \State $\begin{pmatrix}
                x_1 \\ x_2 \\ x_3
            \end{pmatrix}
            \xleftarrow{F_1}
            \begin{pmatrix}
                x_1 \\ 3x_1x_2 - x_3 \\ x_2
            \end{pmatrix}$
        \Else
            \State $\begin{pmatrix}
                x_1 \\ x_2 \\ x_3
            \end{pmatrix}
            \xleftarrow{F_2}
            \begin{pmatrix}
                x_2 \\ 3x_2x_3 - x_1 \\ x_3
            \end{pmatrix}$
        \EndIf
    \EndWhile
}

\rev{In this example, we compute the vector space $I_{\mL,\Q[\xb]_{\leq 3}}$ of all polynomial invariants of degree at most~$3$ for~$\mL$. 
Let $\gb$ denote the set of all monomials of degree up to~$3$. 
We apply \TruncatedIdealBranch\ to the input $((1,1,2),\gb,1,(F_1,F_2))$ in order to compute $I_{\mL,\Q[\xb]_{\leq 3}}$. 
Assume that 
\[
g = b_1 + b_2 x_1 + \dotsb + b_{20} x_3^3
\]
is a polynomial invariant of~$\mL$. 
Since $\mL$ has two branches, there are $32$ linear constraints in $\mathbf{G}_{\leq 5}$.  
By Proposition~\ref{prop_sufficientbranching}, any polynomial invariant of~$\mL$ must satisfy these $32$ linear equations, whose solution space is one-dimensional and generated by
$x_1^2 + x_2^2 + x_3^2 - 3 x_1 x_2 x_3$.
Finally, Algorithm~\CheckPIBranch\ verifies that this polynomial is indeed an invariant of~$\mL$.  
Therefore, $I_{\mL,\Q[\xb]_{\leq 3}}$ has dimension one and is spanned by $x_1^2 + x_2^2 + x_3^2 - 3x_1x_2x_3$.}
\end{example}

\section{General polynomial invariants of specific forms}\label{sec:lift}
%%%%%%%%%%%%%%%%%%%%%%%%%%%%%%%%%%%%%%%%%%%%%%%%%%%%%%%%%%
In this section, we show that when searching for polynomials of a specific form, it is possible to identify those that hold for any initial value much more efficiently than in Section~\ref{sectiongenerating}. Specifically, we first present a necessary and sufficient condition for identifying polynomial invariants where only the constant coefficient depends on the initial value. We then use this condition to detect and generate all polynomial invariants of this form.

Let $h$ and $F = (F_1, \dots, F_n)$ be polynomials in $\mathbb{C}[\mathbf{x}]$. We begin with the following lemma, which reduces a slightly more general form of polynomial invariants to the special case of $f(\mathbf{x}) - f(\mathbf{a})$.

\begin{lemma}\label{lemma:constantcoef}
Let $f$, $g$, and $P$ be non-zero polynomials in $\mathbb{C}[\mathbf{x}]$.~Then,~$P(\mathbf{a}) f(\mathbf{x}) - g(\mathbf{a})$ is in $I_{\mathcal{L}(\mathbf{a}, h, F)}$ for any $\mathbf{a} \in \mathbb{C}^n$ if and only if the following~hold:
\begin{enumerate}
    \item $f(\mathbf{x}) - f(\mathbf{a}) \in I_{\mathcal{L}(\mathbf{a}, h, F)}$ for all $\mathbf{a} \in \mathbb{C}^n$,
    \item $g(\mathbf{x}) = P(\mathbf{x}) f(\mathbf{x})$.
\end{enumerate}
\end{lemma}
\begin{proof}
    %Since $f(\xb)-P(\ab)\in I_{\mL(\ab, h,F)}$ any initial value $\ab \in \Cn$, we have $P(\ab)=f(\ab)$ for any $\ab \in \Cn$. Therefore, $P(\xb)=f(\xb)$.
    
    Assume that $P(\ab)f(\xb)-g(\ab)\in I_{\mL(\ab, h, F)}$ for any $\ab\in \Cn$. It follows that $P(\ab)f(\ab)-g(\ab)=0$ for any $\ab\in \Cn$, that is 
    $g(\xb)=P(\xb)f(\xb).$
    Therefore, $f(\xb)-f(\ab)\in I_{\mL(\ab, h, F)}$ for every initial value $\ab\notin \V(P)$. Choose $\ab\in \V(P)$. Since $\V(P)$ is a hypersurface, there exists a sequence $\{\ab_n\}_{n\in \N}$ such that $\ab_n\notin \V(P)$ for any $n\in \N$ and $\displaystyle\lim_{n\rightarrow \infty}\ab_n= \ab.$  
    Given that $f(\xb)-f(\ab_n)$ is a polynomial invariant, for any natural number $N$, we have $$h(\ab_n)\cdot\ldots\cdot h(F^{N-1}(\ab_n))(f(F^{N}(\ab_n))-f(\ab_n))=0.$$ Since $F,h$ and $f$ are polynomial maps, $F, h$ and $f$ are continuous. Hence,
    \begin{align*}
    &h(\ab)\cdot\ldots\cdot h(F^{N-1}(\ab))\Big(f(F^{N}(\ab))-f(\ab)\Big)\\
    =&\displaystyle\lim_{n\rightarrow \infty}h(\ab_n)\cdot\ldots\cdot h(F^{N-1}(\ab_n))\cdot (f(F^{N}(\ab_n))-f(\ab_n))=0,
    \end{align*}
%for any natural number $N$ 
which implies that $f(\xb)-f(\ab)$ is a polynomial invariant for every $\ab\in \V(P)$. Therefore, $f(\xb)-f(\ab)\in I_{\mL(\ab, h, F)}$ for any $\ab\in \Cn$. The converse is trivial.
\end{proof}    

Therefore, polynomial invariants of the form $f(\mathbf{x}) - f(\mathbf{a})$ encompass those of the form $P(\mathbf{a}) f(\mathbf{x}) - g(\mathbf{a})$, where the constant coefficient is a rational function of the initial value. 
Before generating all polynomial invariants of this form, we first establish a necessary and sufficient condition for identifying them. Then, we use this condition to generate all such polynomial invariants, for any polynomial $f$ within a prescribed finite-dimensional vector space. 

\begin{proposition}\label{general}
    Let $\mathcal{L}(\mathbf{a}, h, F)$ be a polynomial loop, and  $f(\mathbf{x}) \in \mathbb{C}[\mathbf{x}]$. Let $y$ and $z$ be new indeterminates, and $F_1(\mathbf{x}, y, z) = (F(\mathbf{x}), y, z h(\mathbf{x}))$. Define 
$$ X = \V(z(f(\mathbf{x}) - y)) \subset \mathbb{C}^{n+2}. $$
Then, $f(\mathbf{x}) - f(\mathbf{a}) \in I_{\mathcal{L}(\mathbf{a}, h, F)}$ for any $\mathbf{a} \in \mathbb{C}^n$ if and only if $S_{(F_1, X)} = X$.
\end{proposition}
\begin{proof} 
First, assume that $f(\mathbf{x}) - f(\mathbf{a}) \in I_{\mathcal{L}(\mathbf{a}, h, F)}$ for any initial value $\mathbf{a} \in \mathbb{C}^n$. Let $(\mathbf{a}, b, c) \in X$. By definition, one of the following holds:
$$ c = 0 \quad \text{or} \quad b = f(\mathbf{a}). $$
If $c = 0$, then $F_1^m(\mathbf{a}, b, 0) = (F^m(\mathbf{a}), b, 0)$, so $F_1^m(\mathbf{a}, b, 0) \in X$ and $(\mathbf{a}, b, c) \in F^{-m}(X)$ for any $m \in \mathbb{N}$.
If $b = f(\mathbf{a})$, then $f(F^m(\mathbf{a})) = f(\mathbf{a}) = b,$ because $f(\mathbf{x}) - f(\mathbf{a})$ is a polynomial invariant. Thus, $F_1^m(\mathbf{a}, f(\mathbf{a}), c) \in X$, implying that $(\mathbf{a}, b, c) \in F^{-m}(X)$ for any $m \in \mathbb{N}$.
Since
\[
S_{(F_1, X)} = \bigcap_{i=0}^N F_1^{-i}(X) \text{ for some } N \in \mathbb{N}, 
\]
we conclude that $X = \V(z(f(\mathbf{x}) - y)) \subset S_{(F_1, X)}$.
The reverse inclusion holds by definition, so we have $S_{(F_1, X)} = X$.

Conversely, assume that $X = S_{(F_1, X)}$. Given that $(\mathbf{a}, f(\mathbf{a}), 1) \in X$, it follows that $(\mathbf{a}, f(\mathbf{a}), 1) \in S_{(F_1, X)}$. Thus, by Proposition~\ref{prop3.4}, $f(\mathbf{x}) - f(\mathbf{a}) \in I_{\mathcal{L}(\mathbf{a}, h, F)}$ for any initial value $\mathbf{a} \in \mathbb{C}^n$.
\end{proof}

We now express the previous geometric condition as an algebraic one.
\begin{corollary}\label{generalcor}
    Let $\mL(\ab,h,F)$ be a polynomial loop with $h\neq 0$ and  $f\in \C[\xb]$. Then, $f(\xb)-f(\ab)\in I_{\mathcal{L}(\ab,h,F)}$ for any $\ab\in \Cn$ if and only if $f(F(\xb))=f(\xb)$.
\end{corollary}
\begin{proof}
Assume that $f(\xb) - f(\ab) \in I_{\mathcal{L}(\ab, h, F)}$ for any $\ab \in \C^n$. By Proposition~\ref{general}, and using the same notations, we have that $X = S_{(F_1, X)}$.

Let $X_1 = X \cap F_1^{-1}(X)$. By Proposition~\ref{prop:stabilization}.(a), we know that $X \subseteq X_1 \subseteq X$, so it follows that $X = X_1$. 
By the definition of $X_1$, we have:
$$ X_1 = V(z(f(\xb) - y), \; zh(\xb)(f(F(\xb)) - y)). $$
Thus, we can conclude that $V(z(f(\xb) - y)) \subseteq V(zh(\xb)(f(F(\xb)) - y))$. By Hilbert's Nullstellensatz, there exists some $n \in \mathbb{N}$ such that:
$$ z^n h(\xb)^n (f(F(\xb)) - y)^n \in \langle z(f(\xb) - y) \rangle. $$
Note that we can expand the terms as follows:
\begin{align*}
    (zh)^n (f(F(\xb)) - y)^n 
    &= (zh)^n \left[f(\xb) - y + (f(F(\xb)) - f(\xb))\right]^n \\
    &= (zh)^n \left[ A(\xb, y)(f(\xb) - y) + \left(f(\xb) - f(F(\xb))\right)^n \right],
\end{align*}
where $A(\xb, y) \in \C[\xb, y]$. Therefore, it follows that:
$$ (zh)^n (f(\xb) - f(F(\xb)))^n \in \langle z(f(\xb) - y) \rangle. $$
Since $h \neq 0$ and there is no $y$-variable in $(zh)^n (f(\xb) - f(F(\xb)))^n$, we conclude that $f(\xb) = f(F(\xb))$. 
Conversely, if $f(\xb) = f(F(\xb))$, then:
$$ X_1 = V(z(f(\xb) - y), zh(f(F(\xb)) - y)) = V(z(f(\xb) - y)) = X. $$
By Proposition~\ref{prop:stabilization}, we have $X = S_{(F_1, X)}$, and using Proposition~\ref{general}, it follows that $f(\xb) - f(\ab) \in I_{\mathcal{L}(\ab, h, F)}$ for any $\ab \in \C^n$.
\end{proof}

In the following, we extend Corollary~\ref{generalcor} to the case of branching loops.
\begin{corollary}\label{cor:generalbranch}
     Let $F=(F_1,\ldots,F_k)$ and $\mL(\ab,h,F)$ be a branching loop with $h\neq 0$ and  $f(\xb)\in \C[\xb]$. Then, $f(\xb)-f(\ab)\in I_{\mathcal{L}(\ab,h,F)}$ for any $\ab\in \Cn$ if and only if $f(F_i(\xb))=f(\xb)$ for any $i\in \{1,\ldots, k\}$.
\end{corollary}
\begin{proof}
   Assume that $f(\xb)-f(\ab)\in I_{\mL(\ab, h,F)}$ for any $\ab \in \Cn$, and in particular, that $f(\xb)-f(\ab)\in I_{\mL(\ab, h,F_i)}$ for all $i \in \{1, \dots, k\}$. By Corollary~\ref{general}, we have $f(F_i(\xb)) = f(\xb)$ for all $i \in \{1, \dots, k\}$.

Conversely, assume that $f(F_i(\xb)) = f(\xb)$ for all $i$. Define $X_0 = \V(z(f(\xb)-y))$, and for $i \geq 1$, define $F_{i,1}(\xb,y,z)=(F_i(\xb),y, zh(x))$ and let
$$ X_{i+1} = X_i \cap F_{1,1}^{-1}(X_i) \cap \cdots \cap F_{k,1}^{-1}(X_i).$$
Since $z(f(\xb)-y) = z(f(F_i(\xb))-y)$ for all $i \in \{1, \dots, k\}$, we obtain
$$ X_1 = \V\Big(z(f(\xb)-y), z(f(F_1(\xb))-y), \dots, z(f(F_k(\xb))-y)\Big) = X. $$
Therefore, $S_{((F_{1,1}, \dots, F_{k,1}), X)} = \V(z(f(\xb)-y))$, which contains $(\ab, f(\ab), 1)$. By Proposition~\ref{prop:branching}, we conclude that $f(\xb) - f(\ab) \in I_{\mL(\ab, h, F)}$ for all $\ab \in \Cn$.
\end{proof}

To apply the previous corollary, we design the following algorithm to generate a basis for the vector space of all polynomial invariants of the form $f(\xb) - f(\ab)$, where $f$ belongs to a prescribed vector space (e.g., polynomials of bounded degree).
We use the procedure $\coefficients$, which takes as input a polynomial $P \in \C[\xb, \yb]$ and outputs the list of coefficients of $P$, viewed as an element of $\C[\yb][\xb]$, for some arbitrary monomial ordering. For example:
$$ \coefficients\left((y_1 - y_2)x_1^2 + y_1 x_1 x_2\right) = (y_1 - y_2, y_1). $$
The procedure $l_i$ returns the $i$th coordinate of an element in $\Q^n$.

\begin{algorithm}[H]
\caption{%General polynomial invariant
\textsf{SpecificForm}}\label{alg3}
\begin{algorithmic}[1]\setstretch{1.25}
\Require  Sequences of polynomials 
%$\hb=(h_1,\ldots, h_m), 
$F_1,\ldots, F_k$ in $\mathbb{Q}[\xb]$, where $F_i=(F_{i,1},\dotsc,F_{i,n})$ and $\gb = (g_1,\ldots, g_m)$ spanning a vector space $E\subset \C[\xb]$.
\Ensure %Polynomials forming 
A vector space basis for all polynomial invariants of the form $f(\xb)-f(\ab)$, where $f \in E$ and $\ab \in \C^n$, for any loop $\mL(\ab,\hb,(F_1,\ldots, F_k))$ and $\hb\subset \C[\xb]$ has only non-zero entries.
%\State$h=h_1\cdot\ldots\cdot h_m;$
\State$ g\gets y_1g_1(\xb)+\cdots+y_mg_m(\xb);$
\State\label{alg6:2}$(D_1,\ldots, D_k) \gets (g(\xb,\yb)-g(F_1(\xb),\yb),\ldots, g(\xb,\yb)-g(F_k(\xb),\yb));$
\State\label{step4} $(L_1(\yb),\ldots,L_{r}(\yb))\gets \textsf{concatenate}(\coefficients(D_1),\ldots, \coefficients(D_k));$
\State\label{step5} $(c_1,\ldots,c_s)\gets \VectorBasis(L_1(\yb),\ldots,L_{r}(\yb));$
\State\label{alg6:5} \Return {\small $(\displaystyle\sum_{i=1}^m l_i(c_1)g_i(\xb)-\displaystyle\sum_{i=1}^m l_i(c_1)g_i(\ab),\ldots,\displaystyle\sum_{i=1}^m l_i(c_s)g_i(\xb)-\displaystyle\sum_{i=1}^m l_i(c_{s})g_i(\ab));$}
\end{algorithmic}
\end{algorithm}

\begin{theorem}\label{generalinvariants}
  Let $\gb = (g_1, \ldots, g_m)$ and $F_i = (F_{i,1}, \dotsc, F_{i,n})$ be sequences of polynomials in $\Q[\xb]$, for $1 \leq i \leq k$. Let $E$ denote the vector space spanned by $\gb$. Given $\gb$ and $F_1, \ldots, F_k$ as input, Algorithm~\ref{alg3} outputs a sequence of polynomials $\mathbf{\Pb}$ in $\Q[\xb, \ab]$ such that, for any $\ab \in \Cn$ and any $\hb \subset \Q[\xb]$ with all entries non-zero, the following hold.  
\begin{itemize}

\item The polynomials in $\mathbf{\Pb}$ form a basis for the vector space of all polynomial invariants of $\mathcal{L}(\ab, \hb, F)$ of the form $f(\xb) - f(\ab)$, where $f \in E$.

\item Furthermore, if $\dg$ and $\dF[1], \dotsc, \dF[k]$ are bounds on the degrees of $\gb$ and $F_1, \dotsc, F_k$ respectively, the algorithm performs at most  
\[
O\left(\sum_{i=1}^k \dg^n \dF[i]^n \left(\frac{\dg^{2n} \dF[i]^n}{n^{2n}} + m^2\right)\right) = k \cdot m^2 (\dg \dF)^{O(n)}
\]  
operations in $\Q$, where $\dF = \max\{\dF[1], \dotsc, \dF[k]\}$.
\end{itemize}
\end{theorem}
\begin{proof}
Without loss of generality, by taking the product of all non-zero entries of $\hb$, one can assume that $\hb$ reduces to a single non-zero polynomial $h \in \Q[\xb]$.  
Since $g(\xb, \yb)$ is linear in the variables $\yb$, the polynomials $L_1(\yb), \ldots, L_r(\yb)$ are linear. Let $V_{gen}$ denote the vector space of all polynomial invariants of the form $f(\xb) - f(\ab) = 0$, where $f \in E$.  
Suppose Algorithm~\ref{alg3} is executed on the inputs $\gb$ and $F_1, \ldots, F_k$, and outputs
\[
(P_1(\xb) - P_1(\ab), \ldots, P_s(\xb) - P_s(\ab)).
\]  
Therefore, $P_t = \textstyle\sum_{i=1}^m l_i(c_s) g_i(\xb)$, where $g(\xb, c_i) = g(F_j(\xb), c_i)$ for any $t \leq s$ and $j \leq k$. This implies that $P_i(\xb) = P_i(F_j(\xb))$ for any $j \leq k$ and $i \leq t$. Moreover, $P_1, \ldots, P_s$ lie within the vector space $E$.  
By Corollary~\ref{cor:generalbranch}, it follows that $P_i(\xb) - P_i(\ab) \in I_{\mL(\ab, h, (F_1, \ldots, F_k))}$ for any initial value $\ab \in \Cn$ and any $h \neq 0$. Consequently, the vector space $V_{gen}$ includes the space generated by $(P_1(\xb) - P_1(\ab), \ldots, P_s(\xb) - P_s(\ab))$.  

\smallskip
    To prove the converse, let $f(\xb) - f(\ab) \in V_{gen}$ and let $\bb \in \C^m$ such that  
$f = \sum_{i=1}^m l_i(\bb) g_i(\xb)$.
Since $h \neq 0$, by Corollary~\ref{cor:generalbranch}, we know that $f(F_i(\xb)) = f(\xb)$ for all $i \leq k$. Consequently,  
$L_1(\bb) = \cdots = L_r(\bb) = 0$.
Thus, $\bb$ lies in the kernel of the linear forms $L_1, \ldots, L_r$. As a result, there exist scalars $e_1, \ldots, e_s \in \C$ such that  
$\bb = e_1 c_1 + \cdots + e_s c_s$.
Since $l_1, \ldots, l_m$ are linear functions, it follows that:     
\begin{align*}
f(\xb)-f(\ab) &=\displaystyle\sum_{i=1}^m l_i(\bb)g_i(\xb)-\displaystyle\sum_{i=1}^m l_i(\bb)g_i(\ab)\\
&=\displaystyle\sum_{i=1}^m l_i(e_1c_1+\cdots +e_sc_s)g_i(\xb)-\displaystyle\sum_{i=1}^m l_i(e_1c_1+\cdots +e_sc_s)g_i(\ab)\\
&= \displaystyle\sum_{i=1}^m (e_1l_i(c_1)+\cdots +e_sl_i(c_s))g_i(\xb)-\displaystyle\sum_{i=1}^m (e_1l_i(c_1)+\cdots +e_sl_i(c_s))g_i(\ab)\\
&=e_1(P_1(\xb)-P_1(\ab))+\cdots+e_s(P_s(\xb)-P_s(\ab))
\end{align*}
Hence, $V_{gen}$ is contained within the vector space spanned by  
\[
(P_1(\xb) - P_1(\ab), \ldots, P_s(\xb) - P_s(\ab)),
\]  
proving the first part of the theorem, namely the correctness of Algorithm~\ref{alg3}.
We now analyze its complexity. 

The two computationally intensive steps of Algorithm~\ref{alg3} are Steps~\ref{alg6:2} and~\ref{step5}, which involve multivariate polynomial arithmetic and linear algebra, respectively.
At Step~\ref{alg6:2}, for each $1 \leq i \leq k$, $D_i$ can be computed in at most  
$\textstyle{O\left(\frac{\dg^{3n}\dF[i]^{2n}}{n^{2n}}\right)}$
operations in $\Q$, according to \cite{HL2013}.
Step~\ref{step5} involves solving a linear system with $r$ equations and $m$ unknowns (the $y_i$'s). Using Gaussian elimination, this can be done at a cost of $O(m^2r)$. However, from Step~\ref{step5}, the value of $r$ corresponds to the sum of the supports of the $D_i$'s, which satisfies  
$r \leq \sum_{i=1}^k \dg^n\dF[i]^n$.
Combining the complexity bounds for Steps~\ref{alg6:2} and~\ref{step5},  we obtain the total complexity of Algorithm~\ref{alg3} as claimed.
\end{proof}

\begin{remark}
In Algorithm~\ref{alg3}, we restrict our focus to invariants for loops where no zero polynomials appear in the guard. This condition can be verified with negligible computational cost. If the condition is not satisfied, the loop does not iterate since the statement $0 \neq 0$ is trivially false. In such cases, the entire invariant ideal reduces to the maximal ideal  
$\langle x_1 - a_1, \dotsc, x_n - a_n \rangle$.
\end{remark}

\begin{example}
Consider the loop Fib 1:

\programboxappendix[0.5\linewidth]{
\State$(x_1, x_2,x_3)=(2,1,1)$
\While{true}
\State $\begin{pmatrix}
x_1 \\
x_2 \\
x_3
\end{pmatrix}
\longleftarrow
\begin{pmatrix}
x_2\\
x_3\\
2x_2x_3-x_1
\end{pmatrix}
$
\EndWhile
}

\noindent 
We compute all polynomial invariants of the form $f(x_1, x_2, x_3) - f(a_1, a_2, a_3) = 0$ up to degree 4 using Algorithm~\ref{alg3}. Define  
\[
g = y_1 + y_2x_3 + y_3x_2 + y_4x_1 + \cdots + y_{35}x_1^4.
\]  
At Step~\ref{step4}, 54 linear equations are generated, including examples such as  
\[
y_2 - y_3, \quad y_4 - y_3, \quad 16y_{11}, \quad \text{and} \quad y_{10} - y_8.
\]  
At Step~\ref{step5}, the algorithm computes a vector basis for the solution space of this system of 54 linear equations. From the result, we obtain the  output:  
\[
x_1^2 + x_2^2 + x_3^2 - 2x_1x_2x_3 - \big(a_1^2 + a_2^2 + a_3^2 - 2a_1a_2a_3\big).
\]  
Thus, the above polynomial is the only polynomial invariant of the form $f(x_1, x_2, x_3) - f(a_1, a_2, a_3)$ of degree $\leq 4$, up to scalar multiplication.
\end{example}

%%%%%%%%%%%%%%%%%%%%%%%%%%%%%%%%%%%%%%%%%%%%%%%%%%%%%%%%%%%%%%%%%%%%%%%%%%%
\section{%A sufficient criterion for 
Termination of algebraic and semi-algebraic loops }\label{sec:termination}
%%%%%%%%%%%%%%%%%%%%%%%%%%%%%%%%%%%%%%%%%%%%%%%%%%%%%%%%%%%%%%%%%%%%%%%%%%%

In this section, we use invariant sets to establish termination conditions for algebraic and semi-algebraic loops. We begin by presenting a necessary and sufficient condition for the termination of algebraic loops. In contrast, for semi-algebraic loops, we provide only a sufficient condition. %for termination.

\begin{proposition}\label{prop:terminateinvariant}
    Let $\ab\in\mathbb{C}^n$, $\gb = (g_1,\ldots, g_m)$  and $F=(f_1,\ldots, f_n)$ be  two sequences of polynomials in $\C[\xb]$. Let $X=\V(g_1,\ldots, g_m)\subset \Cn$.
    Then, the polynomial loop $\mathcal{L}(\ab,\gb=0,F)$ never terminates if and only if $\ab \in S_{(F, X)}$.
\end{proposition}
\begin{proof} 
% Let $\ab= (a_1, a_2,\ldots, a_n)$ be the initial value of $\mathcal{L}$. 
The statement follows directly from the definition, since the loop $\mathcal{L}(\ab, X, F)$ does not terminate if and only if $F^{(m)}(\ab) \in X$ for all $m \geq 0$. This is equivalent to $\ab \in S_{(F, X)}$.
\end{proof}

\begin{example} %\textcolor{red}{Let's} $\rightarrow$ \textcolor{blue}{correct: Let us}
Consider the loop $\mL(\ab, g,F)$ from Example~\ref{example1}.
\begin{comment}
    \programbox[0.53\linewidth]{
\State$(x_1, x_2)=(a_1,a_2)$
\While{$g=0$}
\State $\begin{pmatrix}
x_1 \\
x_2
\end{pmatrix}
\xleftarrow{F}
\begin{pmatrix}
10x_1-8x_2\\
6x_1-4x_2
\end{pmatrix}
$
\EndWhile
}
\\
\end{comment}
As we have seen in Example~\ref{example1}, the output of $\InvariantSet(g,F)$ is $(g,g\circ F)$. Therefore, $\mathcal{L}$ never terminates if and only if $(a_1, a_2)\in \V(g, g\circ F)$.
\end{example}

%\vspace{-3mm}
\begin{definition}\label{def:saloop}
Consider the basic semi-algebraic set $S$ of  $\R^n$ defined by $g_1=\cdots=g_k=0$ and $h_1>0,\ldots,h_s>0$ and a polynomial map $F= (f_1,\ldots,f_n)$, where the $f_i$'s, the $g_j$'s  and the $h_j$'s are polynomials in $\mathbb{R}[\xb]$.
Then a loop of the form:

\programbox[0.75\linewidth]{
\State$(x_1, x_2,\ldots, x_n)=(a_1,a_2,\ldots,a_n)$
\While{$g_1 = \cdots = g_k =0\textbf{ and } h_1>0,\dotsc,h_s>0$}
\State $\begin{pmatrix}
x_1 \\
x_2 \\
\vdots \\
x_n
\end{pmatrix}
\xleftarrow{F}
\begin{pmatrix}
f_1\\
f_2\\
\vdots\\
f_n
\end{pmatrix}
$
\EndWhile
}

\noindent is called a \emph{semi-algebraic loop} on $S$ with respect to~$F$. 
We denote by $\mathcal{S}(\gb, \hb)$ the solution set in $\mathbb{R}^n$ of the polynomial system defined by $\gb$ and $\hb$.
\end{definition}

The following proposition is a direct consequence of the definitions.
\begin{proposition}\label{prop:saloop}
Let $\ab \in \R^n$, and $\gb$ and $\hb$ be as above. %in Definition~\ref{def:saloop}.
Let $r_1,\dotsc,r_p$ be polynomial invariants of $\mathcal{L}(\ab,0,F)$. Then
the semi-algebraic loop $\mathcal{L}(\ab,(\gb,\hb),F)$ never terminates if
$
    \V(r_1,\dotsc,r_p) \cap \R^n \subset \mathcal{S}(\gb,\hb).
$
\end{proposition}
The above inclusion corresponds to the quantified formula:
\[
\forall \xb \in \R^n, r_1(\xb)=\cdots=r_p(\xb)=0 \Rightarrow\left\{\hspace*{-0.2cm}\begin{array}{l}
g_1(\xb) = \cdots = g_k(\xb) =0\\h_1(\xb)>0,\dotsc,h_s(\xb)>0\end{array}\right..
\]
The validity of such a formula can be determined using a quantifier elimination algorithm \cite[Chapter 14]{bpr2006}. Since the formula contains no free variables or alternating quantifiers, it corresponds to the emptiness problem of the solution set of a system of polynomial equations and inequalities. This can be efficiently addressed by specialized algorithms, with the most general version presented in \cite[Theorem 13.24]{bpr2006}. Furthermore, given the specific structure of the formula, a more efficient approach may involve the method outlined in \cite{synthesis2023algebro}, which combines the Real Nullstellensatz \cite{BCR1998} and Putinar's Positivstellensatz \cite{Pu1993}.
We will not explore these aspects further, as they fall outside the scope of this paper and will be addressed in future work. Instead, we focus on showing why the above sufficient criterion is not necessary.
\begin{example}\label{exsaloop}
Consider the elementary semi-algebraic loop:

\programboxappendix[0.45\linewidth]{
\State$(x_1, x_2)=(a_1,a_2)$
\While{$x_1>0$}
\State $\begin{pmatrix}
x_1 \\
x_2
\end{pmatrix}
\longleftarrow
\begin{pmatrix}
2x_1\\
2x_2
\end{pmatrix}
$
\EndWhile
}

A direct analysis of the linear recursive sequence defined by the successive values $\ab^0, \ab^1, \dotsc$ of $(x_1, x_2)$ reveals that the loop never terminates if and only if $a_1 > 0$. Moreover, the polynomial $a_2x_1 - a_1x_2 = 0$ is an invariant of this loop. Since every $\ab^j$, for $j \geq 0$, must lie on this line, it generates the entire invariant ideal. However, $\V(a_2x_1 - a_1x_2) \cap \mathbb{R}^2$ is not contained within $\mathcal{S}(0, x_1)$, as illustrated in Figure~\ref{fig:saloop}.
\end{example}
\vspace{-4mm}
\begin{figure}[H]\centering
\includegraphics[width=0.35\linewidth]{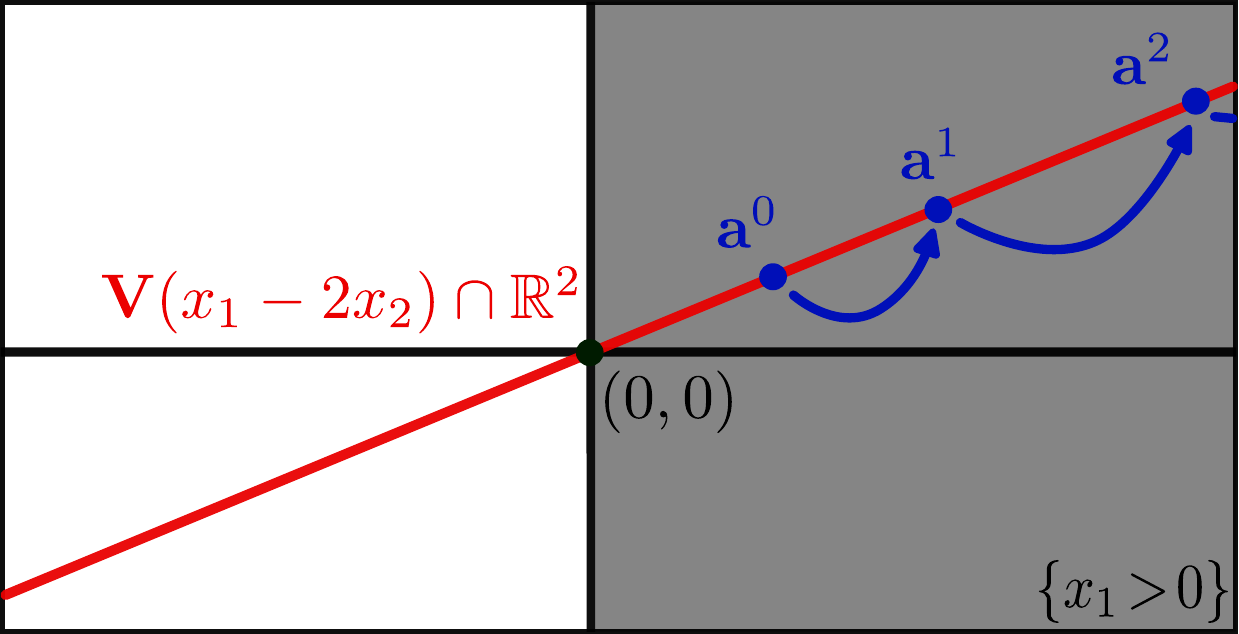}
\caption{\textmd{An illustration of a particular case of Example~\ref{exsaloop} for $(a_1,a_2)=(2,1)$. In blue are depicted the successive values $\ab^0,\ab^1,\dotsc$ of the variables $(x_1,x_2)$, in red is the real zero-set of the invariant ideal, and in gray the set $\mathcal{S}(0,x_1)$ defined by the condition $x_1>0$.}}\label{fig:saloop}
\end{figure}

\vspace{-7mm}
%%%%%%%%%%%%%%%%%%%%%%%%%%%%%%%%%%%%%%%%%%%%%%%%%%%%%%%%%%
%%%%%%%%%%%%%%%%%%%%%%%%%%%%%%%%%%%%%%%%%%%%%%%%%%%%%%%%%%
\section{Implementation and Experiments
%Implementation and experimental results
}\label{sec:implementation}
%%%%%%%%%%%%%%%%%%%%%%%%%%%%%%%%%%%%%%%%%%%%%%%%%%%%%%%%%%
In this section, we present an implementation of the algorithms discussed in this paper and compare its performance with \textsf{Polar}~\cite{DBLP:journals/pacmpl/MoosbruggerSBK22}, which primarily builds on \cite{Unsolvableloops} for the case of unsolvable loops. In all tables, $n$ represents the number of program variables, and $D$ denotes the degree of the map of a loop. To ensure a fair comparison, and in line with the existing literature, all experiments compute polynomial invariants with a degree bounded by a parameter that we vary. This choice is facilitated by the flexibility of our algorithms in terms of the types of invariants they can compute.
%\textcolor{red}{Example~\ref{ex2} is sourced from \cite{hrushovski2018polynomial}, ``Floor function'' is cited from \cite{rodriguez2004automatic}, while the remaining examples are from~\cite{amrollahi2023solvable}.}
The implementation of Algorithms~\ref{matrixalgo}, \ref{alg2}, and \ref{alg3} in Macaulay2~\cite{M2} is available at: % the following link:
\begin{center}\footnotesize
\href{https://github.com/FatemehMohammadi/Algebraic_PolyLoop_Invariants.git}{\texttt{github.com/FatemehMohammadi/Algebraic\_PolyLoop\_Invariants.git}}
\end{center}
%%%%%%%%%%%%%%%%%%%%%%%%%%%%%%%%%%%%%%%%%%%%%%%%%%%%%%%%%%
\subsection{Implementation details}
The experiments below were performed on a laptop featuring a 4.8 GHz Intel i7 processor, 16 GB of RAM, and a 25 MB L3 cache. The implementation primarily relies on standard linear algebra routines and Gröbner basis computations provided by Macaulay2. This approach ensures that the implementation closely follows the pseudo-code presented in this paper.
We made minor modifications to these algorithms to speed up computations, based on insights gained from experimental observations, detailed below.

We observed that for most polynomial loops, all candidate polynomials in $\mathcal{B}$, computed at Step~\ref{step:B} of Algorithm~\ref{alg2}, are polynomial invariants. 
\rev{Although Algorithm~\CheckPI\ is designed to verify the invariance of a single polynomial, it can naturally be extended to check several candidate invariants $g_1,\dots,g_k$ simultaneously. 
To do so, it suffices to replace $\InvariantSet(zg,F_0)$ with $\InvariantSet((zg_1,\dots,zg_k),F_0)$ in Step~\ref{CheckPIStep3} of Algorithm~\CheckPI. 
This modification accelerates the verification process, since the stabilization in \InvariantSet\ typically occurs earlier.}
Therefore, instead of checking each polynomial individually, we verify all elements of $\mathcal{B}$ at the same time. 
This reduces the dimension of the variety that must be considered and can yield speedups of up to two orders of magnitude (as illustrated in Example~\ref{ex3}).

\subsection{Experimental results}
In Tables~\ref{table 1} and~\ref{table2}, we compare our implementation of Algorithm~\ref{alg3} with the software \textsf{Polar}, which is based on the methods of~\cite{Unsolvableloops, amrollahi2025solvable} for handling unsolvable loops. 
On these benchmarks, \textsf{Polar} only produces invariants of the form $f(\xb)-f(\ab)$, where $\ab$ is the initial value and $f \in \Q[\xb]$. 
In contrast, Algorithm~\ref{alg3} can generate all polynomial invariants of this form whenever they exist. 
\rev{The benchmarks are taken from~\cite{amrollahi2025solvable} and are available at}
\begin{center}\footnotesize
\href{https://github.com/FatemehMohammadi/Algebraic_PolyLoop_Invariants/tree/main/software/loops}{\texttt{github.com/FatemehMohammadi/Algebraic\_PolyLoop\_Invariants/software/loops}}
\end{center}
A key distinction between the two approaches is that our method is global: we compute \emph{all} polynomial invariants up to a specified degree, whereas \textsf{Polar} returns only a (possibly empty) \emph{subset} of them. 
However, \textsf{Polar} can treat probabilistic loops, while our approach is restricted to deterministic ones.

\rev{
To compare our results with those of \textsf{Polar} in a meaningful way, we distinguish between \emph{effective} and \emph{non-effective} (or \emph{defective}) program variables. 
A variable is effective if it does not participate in any cyclic or nonlinear dependency within the loop; in such cases its evolution can be solved explicitly, and classical methods (such as Gr\"obner basis techniques) suffice to compute all polynomial invariants involving only effective variables. 
Invariants depending solely on effective variables are therefore not the focus of \textsf{Polar}'s unsolvable-loop analysis. 
For this reason, we report only polynomial invariants that involve at least one non-effective variable. 
This criterion is applied uniformly to the outputs of both \textsf{Polar} and Algorithm~\ref{alg3}.}

In Table~\ref{table 1}, we present quantitative data comparing Algorithm~\ref{alg3} and \textsf{Polar} across several benchmarks (listed in the rows) for generating polynomial invariants of degrees $1,2,3,\dots,8$.  
The column ``Alg~\ref{alg3}" reports the number of polynomial invariants \rev{involving at least one non-effective variable} generated by Algorithm~\ref{alg3}.  
The column ``\textsf{Polar}" reports the number of polynomial invariants \rev{involving at least one non-effective variable} computed by \textsf{Polar}.

\rev{In particular, even in cases where the results match those of \textsf{Polar} (that is, when the value of $d$ in the \textsf{Polar} column coincides with ours), our method can additionally certify that no further linearly independent polynomial invariants exist, whereas \textsf{Polar} cannot. 
Furthermore, for all benchmarks, \textsf{Polar} reaches the time limit earlier than our method.}

\begin{table}[H]
    \centering
    \scalebox{0.52}{
\begin{tabular}{|*{19}{c|}}
    \hline
        \multicolumn{3}{|c}{Degree} & \multicolumn{2}{|c}{1} & \multicolumn{2}{|c}{2} &  \multicolumn{2}{|c}{3} & \multicolumn{2}{|c|}{4}& \multicolumn{2}{|c|}{5}& \multicolumn{2}{|c|}{6}& \multicolumn{2}{|c|}{7}& \multicolumn{2}{|c|}{8} \\ \hline
        Benchmark & $n$ & $D$ & Alg~\ref{alg3} & \textsf{Polar}& Alg~\ref{alg3} & \textsf{Polar}& Alg~\ref{alg3} & \textsf{Polar}& Alg~\ref{alg3} & \textsf{Polar} & Alg~\ref{alg3} & \textsf{Polar}& Alg~\ref{alg3} & \textsf{Polar}& Alg~\ref{alg3} & \textsf{Polar}& Alg~\ref{alg3} & \textsf{Polar}  \\ \hline
       
         Fib1 &3&2 &0 &0 & 0&0   &1&1&1&1&1&1&2&2&2&2&2&2 \\ \hline
        Fib2 &3&3 & 0&0 & 0  &0 &1&1&1&1&1&1&2&2&2&2&2&TL \\ \hline
        Fib3  &3&2  &0&0  &0&0 &1&1&1&1&1&1&2&2&2&TL&2&TL \\ \hline
        Yagzhev9 &9&3 &0&0   &3 &3 &4 &4  &\textbf{10}&TL&TL&TL&TL&TL&TL&TL&TL&TL \\ \hline
        Yagzhev11 &11&3 &0&0   &0&0  &1&1  &TL&TL&TL&TL&TL&TL&TL&TL&TL&TL \\ \hline
        Ex 9 &3&2 &0&0   &0&0 &1&1 & 1&1& 1&1& \textbf{2}&TL& \textbf{3}&TL& \textbf{3}&TL \\ \hline 
        Ex 10 &3 &2&0&0  &0&0  &0&0&0&TL&0&TL&0&TL&0&TL&0&TL \\ \hline
    \parbox{0.8cm}{\centering Markov\\[-0.4em]Triples\\[-0.4em](Ex~\ref{ex: markov})} &3&2 & 0 &0 & 0&0 &  1&1& 1&1& 1&1& 2&2& 2&2& 2&2\\ \hline
    Nagata &3&5 &0&0 &  1&1  & 2&2 &  4&4&  \textbf{6}&TL&  \textbf{9}&TL&  \textbf{12}&TL&  \textbf{16}&TL \\ \hline
        \parbox{0.8cm}{\centering Squares\\[-0.4em](Ex~\ref{ex3})}  &3&2 & 0 &0 & 0&0 &  0&0& 0&0& 0&0& 0&0&0&0&0 &0\\ \hline
\end{tabular}

}

{\small TL = Timeout (360 seconds); \quad
\textbf{bold}: new invariants found}
\caption{\label{table 1}Data on outputs of Algorithm~\ref{alg3} and \textsf{Polar}
}
\end{table}

In Table~\ref{table2}, we present the execution times for Algorithm~\ref{alg3} and \textsf{Polar} (corresponding to the benchmarks of Table~\ref{table 1}), using a time limit of 360 seconds. 
In cases where \textsf{Polar} hits this limit (for example, degree~4 for \texttt{Yaghzev9}), we verified that it does not terminate within 15 minutes and eventually reaches its maximum recursion depth.
\begin{table}[H]
    \centering\hspace*{-0.2cm}
    \scalebox{0.52}{
    \begin{tabular}{|*{19}{c|}}
    \hline
        \multicolumn{3}{|c}{Degree} & \multicolumn{2}{|c}{1} & \multicolumn{2}{|c}{2} &  \multicolumn{2}{|c}{3} & \multicolumn{2}{|c|}{4}& \multicolumn{2}{|c|}{5}& \multicolumn{2}{|c|}{6}& \multicolumn{2}{|c|}{7}& \multicolumn{2}{|c|}{8} \\ \hline
        Benchmark&$n$ &$D$ & Ours &\textsf{Polar} & Ours & \textsf{Polar} & Ours & \textsf{Polar} & Ours & \textsf{Polar}& Ours & \textsf{Polar}& Ours & \textsf{Polar}& Ours & \textsf{Polar}& Ours & \textsf{Polar} \\ \hline      
        Fib1 &3&2& \textbf{0.03} & 0.2 & \textbf{0.04} & 0.32 & \textbf{0.09} & 0.68 & \textbf{0.18} & 1.58& \textbf{0.35} & 3.5& \textbf{0.61} & 16.6& \textbf{1.42} & 67.5& \textbf{2.88} & 308 \\ \hline
        Fib2&3&3 & \textbf{0.03} & 0.23 & \textbf{0.04} & 0.46 & \textbf{0.07} & 1.18 & \textbf{0.15} & 3.69& \textbf{0.34} & 11.5& \textbf{0.65} & 45.8& \textbf{1.25} & 260& \textbf{2.51} &TL\\ \hline
        Fib3&3&2 & \textbf{0.03} & 0.21 & \textbf{0.05} & 0.4 & \textbf{0.08} & 1.26 & \textbf{0.15} & 4.3& \textbf{0.27} & 31.7& \textbf{0.6} & 107.9& \textbf{1.39} &TL& \textbf{2.89} &TL \\ \hline
       
        Yagzhev9&9&3 & \textbf{0.05} & 0.43 & \textbf{0.36} & 5.2 & \textbf{5.7} & 131.5 & \textbf{143.7} & {TL}& TL & {TL}& TL & {TL}& TL & {TL}& TL & {TL} \\ \hline
        Yagzhev11&11&3 & \textbf{0.1} & 0.45 & \textbf{1.1} & 6.83 & \textbf{19.4} & 359 & {TL}& {TL}& {TL}& {TL}& {TL}& {TL}& {TL}& {TL}& {TL}& {TL}\\ \hline
        Ex 9&3&2& \textbf{0.03} & 0.28 & \textbf{0.06} & 0.64 & \textbf{0.1} & 2.38 & \textbf{0.18} & 11.5& \textbf{0.35} & 172& \textbf{0.77}& TL& \textbf{1.6} & TL& \textbf{4.87} & TL \\ \hline
        Ex 10&3&2 & \textbf{0.02} & 0.39 & \textbf{0.05} & 1.7 & \textbf{0.09} & 14.9 & \textbf{0.2} & TL & \textbf{0.38} & TL & \textbf{
        0.98}& TL & \textbf{2.1} & TL & \textbf{8.5} & TL \\ \hline
        \parbox{0.8cm}{\centering Markov\\[-0.4em]Triples\\[-0.4em](Ex~\ref{ex: markov})}&3&2& \textbf{0.04}  & 0.27& \textbf{0.06} & 0.54 & \textbf{0.13}
        & 1.31 & \textbf{0.26} & 2.82& \textbf{0.51} & 6& \textbf{1.3} & 14.87&
       \textbf{2.6}& 33.84& \textbf{4.55}& 88.1 \\ \hline
         Nagata&3&5& \textbf{0.03} & 0.25 & \textbf{0.04} & 1.7 & \textbf{0.08} & 3.78   & \textbf{0.17} & 26.69& \textbf{0.42} & TL& \textbf{0.8} & TL& \textbf{1.97} & TL& \textbf{3.8} & TL \\ \hline
         \parbox{0.8cm}{\centering Squares\\[-0.4em](Ex~\ref{ex3})}&3&2
        & \textbf{0.03}  & 0.5& \textbf{0.04} & 0.67 & \textbf{0.07}
        & 1.15 & \textbf{0.16} & 2.25& \textbf{0.3} & 5.46& \textbf{0.72} & 10.1& \textbf{1.6} & 70.3& \textbf{4.3} & 165.9\\ \hline
    \end{tabular}
    }
    
    {\small TL = Timeout   (360 seconds); }
    \caption{\label{table2}Timings for Algorithm~\ref{alg3} and \textsf{Polar}, in seconds
    }
\end{table}

We observe that our implementation is typically at least an order of magnitude faster than \textsf{Polar}. 
For all benchmarks except \texttt{Yaghzev9} and \texttt{Yaghzev11}, Algorithm~\ref{alg3} computes all polynomial invariants up to degree~10 within 120 seconds, \rev{although it does not terminate within 360 seconds for degrees~$\geq 15$ on any benchmark}.  
\rev{For \texttt{Yaghzev9} and \texttt{Yaghzev11}, the larger number of variables has a noticeable impact: the maximum computable degree drops to~5. This behavior is consistent with the complexity estimate, which is exponential in the number of variables (Theorem~\ref{generalinvariants}).}

Finally, in Table~\ref{table3}, we report experimental results for Algorithm~\ref{alg2} on several benchmarks, using different degrees for the polynomial invariants. 
Since Algorithm~\ref{alg2} requires a fixed initial value, we selected random integers in the range $[-100,100]$ and averaged the running times over five runs. 
For each benchmark, we list the dimension of the vector space of all computed polynomial invariants together with the corresponding average runtime.

\begin{table}[H]
    \centering\hspace*{-0.2cm}
    \scalebox{0.65}{
    \begin{tabular}{|*{15}{c|}}
    \hline
        \multicolumn{3}{|c}{Degree} & \multicolumn{2}{|c}{1} & \multicolumn{2}{|c}{2} &  \multicolumn{2}{|c}{3} & \multicolumn{2}{|c|}{4} & \multicolumn{2}{|c|}{5} & \multicolumn{2}{|c|}{6} \\ \hline
    Benchmark &$n$ & $D$& Timing & dim & Timing & dim & Timing & dim & Timing & dim & Timing & dim & Timing & dim \\ \hline      
        Fib1 & 3&2 & {0.026} & 0 & {0.062} & 0 & {0.34} & 1 & {30.79} & 4 &TL&TL&TL&TL \\ \hline
        Fib2 & 3&3 & {0.019} & 0 & {0.055} & 0 & {25.26} & 1 & {TL} & TL&TL&TL&TL&TL\\ \hline
        Fib3 & 3&2 & 0.017 & 0 & {0.057} & 0 & {3.1} & 1 & {25.12} & 4&TL&TL&TL&TL \\ \hline
       
        Yagzhev9 & 9&3 & {0.089} & 3 & TL & TL & {TL} & TL & {TL} & {TL}&TL&TL&TL&TL \\ \hline
        Yagzhev11 & 11&3 & {0.11} & 0 & {2.64} & 0 & {318} & 1 & {TL}& {TL}&TL&TL&TL&TL\\ \hline
        Ex 9& 3&2 & {0.02} & 0 & {0.055} & 0 & {0.17} & 3 & {0.57} & 11&2.63&25&10.14&46\\ \hline
        Ex 10& 3&2  & {0.017} & 0 & {0.066} & 2 & {0.18} & 8 & {0.62} & 19&2.61&36&11.4&60 \\ \hline
          
        \parbox{0.8cm}{\centering Markov\\[-0.4em]Triples\\[-0.4em](Ex~\ref{ex: markov})}& 3&2 & {0.033}  & 0& {0.11} & 0 & {0.36}
        & 1 & {1.28} & 4 &2.87&10&8.7&20\\ \hline
         Nagata& 3&5 & 0.019 & 1 & {0.057} & 5 & {0.14} & 13 & {0.39} & 26 &0.98&45&2.56&71\\ \hline
         \parbox{0.8cm}{\centering Squares\\[-0.4em](Ex~\ref{ex3})}
      & 3&2   & {0.017}  & 0 & {0.063} & 1 & {0.82}
        & 4 & {TL} & TL &TL&TL&TL&TL\\ \hline
    \end{tabular}
    }
    
    {\small TL = Timeout   (360 seconds); }
    \caption{\label{table3}Timings in seconds and data on outputs for Algorithm~\ref{alg2}  
    }
\end{table}
For Ex9, Ex10, MarkovTriples, and Nagata, Algorithm~\ref{alg2} terminates within 360 seconds for degrees $\leq 8$, \rev{but not for degrees $\geq 9$}. 
For the remaining benchmarks, it fails to terminate within 360 seconds for degrees $\geq 5$. 
Table~\ref{table3} also shows that, for many benchmarks, there are no polynomial invariants of degree $\leq 2$, requiring the search for degree 3 invariants, which are sufficient for all benchmarks.

\rev{We note that when initial values are fixed, our method can find more invariants than \textsf{Polar}.
%for some benchmarks.
%ever non-zero polynomial invariants exist, our method typically finds more than \textsf{Polar}. 
For instance, for \texttt{Squares}, we identify the unique linear invariant
$1 + x_1 + x_2 + x_3 = 0$
(see also Example~\ref{ex3}). 
For \texttt{Yaghzev9} with initial value $(2,-3,1,4,-1,7,-4,3,2)$, we obtain the invariants
\[
x_1 - x_3 + x_5 = 0,\quad
x_2 - x_4 + x_6 = 0,\quad
x_8 - x_7 - 7 = 0.
\]
Furthermore, in the general case, for \texttt{Ex10}, \textsf{Polar} does not detect the following invariant in terms of the initial values $(a_1,a_2,a_3)$, which is computed by Algorithm~\ref{algo:class}:
\begin{flushleft}
$(3a_1-a_2-4a_3)^2(x_1+x_2)-(3a_1-a_2-4a_3)^2(x_2+x_3)-9(a_1-a_3)(x_1+x_2)^2-16(a_1-a_3)(x_2+x_3)^2+   24(a_1-a_3)(x_1+x_2)(x_2+x_3)=0$.
\end{flushleft}
}

When the initial values are fixed, all polynomial invariants of a loop up to a specified degree can be computed as the kernel of the polynomial matrix evaluated at these initial values, as output by Algorithm~\ref{matrixalgo}. However, for most of the examples above, Algorithm~\ref{matrixalgo} does not terminate within an hour.

\section{Conclusion}
\rev{
We presented a new method for generating all polynomial invariants in a given vector subspace $E \subseteq \mathbb{Q}[\xb]$ for branching loops with nondeterministic conditional statements. 
Our main contributions are:
\begin{itemize}
\item Algorithms~\ref{matrixalgo} and~\ref{algo:class} for loops without fixed initial values;
\item Algorithm~\ref{alg2} for loops with specified initial values;
\item Algorithm~\ref{alg3} for computing all invariants of the form $f(\xb)-f(\ab)$.
\end{itemize}
Among these, Algorithm~\ref{algo:class} provides an explicit symbolic description of the invariants, in contrast to Algorithm~\ref{matrixalgo}. 
Although polynomial invariants are not always of the form $f(\xb)-f(\ab)$, Algorithm~\ref{alg3} is considerably faster than both Algorithm~\ref{algo:class} and Algorithm~\ref{alg2}, since it relies solely on linear algebra. 
For loops with a fixed initial value, our method also finds more polynomial invariants than \textsf{Polar}. 
Taken together, these results demonstrate that our approach offers a substantial advance in computing polynomial invariants for unsolvable loops.
}

\rev{
Our work also has a few limitations. In particular, Algorithm~\ref{alg3} handles only invariants of the form $f(\xb)-f(\ab)$, and the general case in Algorithm~\ref{matrixalgo} may become costly when the dimension of the polynomial space grows.
Our work opens several directions for future research. 
In Algorithm~\ref{matrixalgo}, most variables appear only linearly in the coefficient polynomials. 
By exploiting this structure, we could optimize Algorithm~\ref{algo1} for such structured polynomials, which would in turn improve the efficiency of Algorithm~\ref{matrixalgo}.
It would be interesting to generalize the special form $f(\xb)-f(\ab)$ by identifying algebraic conditions under which such invariants can be computed.
}

\color{black}

\medskip\noindent{\bf Acknowledgement.} 
\rev{We thank the anonymous reviewers for their helpful comments and suggestions.
We also thank George Kenison for insightful discussions on \textsf{Polar}.
The authors were partially supported by the KU Leuven grant iBOF/23/064,
the FWO grants G0F5921N and G023721N, and the UiT Aurora project MASCOT.}

\vspace{-2mm}
\bibliographystyle{alpha}
\bibliography{ISAAC_References}
\end{document}